\documentclass[journal]{IEEEtran}
\usepackage{todonotes}

\usepackage{amsmath}
\usepackage{amssymb}
\usepackage{amsthm}
\usepackage{enumitem}
\usepackage{latexsym}
\usepackage{verbatim}
\usepackage{graphicx}
\usepackage{subcaption}
\usepackage{float}

\newtheorem{theorem}{Theorem}

\newtheorem{lemma}{Lemma}

\newtheorem{problem}{Problem}

{\begin{list}               %    with flush left bullets.
    {$\bullet$ \hfill}{
        \setlength{\leftmargin}{\parindent}
        \setlength{\parsep}{0.04\baselineskip}
        \setlength{\itemsep}{0.5\parsep}
        \setlength{\labelwidth}{\leftmargin}
        \setlength{\labelsep}{0em}}
    }
{\end{list}}

\providecommand{\cref}[1]{Chapter~\ref{chap:#1}}

\providecommand{\R}{\ensuremath{\mathbb{R}}}

\providecommand{\inprod}[1]{\left\langle#1\right\rangle}
\providecommand{\set}[1]{\left\{#1\right\}}

\providecommand{\bydef}{\overset{\text{def}}{=}}
\providecommand{\diag}{\mathop{\mathrm{diag}}}
\providecommand{\rank}{\mathop{\mathrm{rank}}}

\renewcommand{\vec}[1]{\ensuremath{\mathbf{#1}}}
\providecommand{\mat}[1]{\ensuremath{\mathbf{#1}}}

\providecommand{\calR}{\mathcal{R}}

 \providecommand{\mD}{\mat{D}}
 
\providecommand{\mI}{\mat{I}}  
  
\providecommand{\mN}{\mat{N}}
\providecommand{\mM}{\mat{M}} \providecommand{\mP}{\mat{P}} 
\providecommand{\mQ}{\mat{Q}} \providecommand{\mR}{\mat{R}}
 \providecommand{\mU}{\mat{U}} 
 \providecommand{\mT}{\mat{T}}

\providecommand{\vc}{\vec{c}}

 \providecommand{\vn}{\vec{n}} 
 \providecommand{\vp}{\vec{p}}
\providecommand{\vq}{\vec{q}} \providecommand{\vr}{\vec{r}}

\providecommand{\vx}{\vec{x}}

 \providecommand{\vv}{\vec{v}}

\newcommand{\RNum}[1]{\uppercase\expandafter{\romannumeral #1\relax}}

\makeatletter
\newcommand*\rel@kern[1]{\kern#1\dimexpr\macc@kerna}
\newcommand*\widebar[1]{%
  \begingroup
  \def\mathaccent##1##2{%
    \rel@kern{0.8}%
    \overline{\rel@kern{-0.8}\macc@nucleus\rel@kern{0.2}}%
    \rel@kern{-0.2}%
  }%
  \macc@depth\@ne
  \let\math@bgroup\@empty \let\math@egroup\macc@set@skewchar
  \mathsurround\z@ \frozen@everymath{\mathgroup\macc@group\relax}%
  \macc@set@skewchar\relax
  \let\mathaccentV\macc@nested@a
  \macc@nested@a\relax111{#1}%
  \endgroup
}
\makeatother

\IEEEoverridecommandlockouts

\ifCLASSINFOpdf
\else
\fi

\begin{document}
\title{Shapes from Echoes: Uniqueness from Point-to-Plane Distance Matrices\\
\thanks{
M. Krekovi\'c and M. Vetterli are with the School of Computer and Communication Sciences, Ecole Polytechnique F\'ed\'erale de Lausanne (EPFL), CH-1015 Lausanne, Switzerland (e-mail: miranda.krekovic@epfl.ch, martin@vetterli@epfl.ch).
I. Dokmani\'c is with the Coordinated Science Lab, University of Illinois at Urbana-Champaign, Urbana, IL 61801, USA (e-mail: dokmanic@illinois.edu).
This work was supported by the Swiss National Science Foundation grant number 20FP-1 151073, ``Inverse Problems regularized by Sparsity''. I. Dokmani\'c was supported by a Google Faculty Research Award.
}
}%

\author{
Miranda~Krekovi\'c,~\IEEEmembership{Student~Member,~IEEE,} Ivan Dokmani\'c,~\IEEEmembership{Member,~IEEE,}
        and~Martin~Vetterli,~\IEEEmembership{Fellow,~IEEE}
}


\markboth{IEEE Transactions on Signal Processing,~Vol.~XX, No.~X, February~2019}%
{Krekovi'\c \MakeLowercase{\textit{et al.}}: Shapes from Echoes: Uniqueness from Point-to-Plane Distance Matrices}



\renewcommand{\arccos}{\mathrm{acos~}}
\renewcommand{\arctan}{\mathrm{atan~}}
\renewcommand{\arcsin}{\mathrm{asin~}}

\maketitle

\begin{abstract}
We study the problem of localizing a configuration of points and planes from the collection of point-to-plane distances. This problem models simultaneous localization and mapping from acoustic echoes as well as the notable ``structure from sound'' approach to microphone localization with unknown sources. In our earlier work we proposed computational methods for localization from point-to-plane distances and noted that such localization suffers from various ambiguities beyond the usual rigid body motions; in this paper we provide a complete characterization of uniqueness. We enumerate equivalence classes of configurations which lead to the same distance measurements as a function of the number of planes and points, and algebraically characterize the related transformations in both 2D and 3D. Here we only discuss uniqueness; computational tools and heuristics for practical localization from point-to-plane distances using sound will be addressed in a companion paper.
\end{abstract}


% \fi
\IEEEpeerreviewmaketitle

\section{Introduction}
\label{introduction}

Localization methods are traditionally based on geometric information (angles, distances, or both) about known objects, often referred to as landmarks or anchors. Famous examples include global positioning by measuring distances to satellites and navigation at sea by measuring angles of celestial bodies. More recent work on simultaneous localization and mapping (SLAM) addresses the case where the positions of landmarks are also unknown.

In this paper, we address localization from distances to (unknown) planes instead of the more extensively studied localization from distances to points. Concretely, given pairwise distances between a set of points and a set of planes, we wish to localize both the planes and the points.
It is clear that a single point does not allow unique localization. As we will show, localization is in general possible with multiple points, though there are surprising exceptions.

Localization from point-to-plane distances models many practical problems. Our motivation comes from indoor localization with sound. Imagine a mobile device equipped with a single omnidirectional source and a single omnidirectional receiver that measures its distance to the surrounding reflectors, for example by emitting acoustic pulses and receiving echoes. The times of flight of the first-order echoes recorded by the device correspond to point-to-plane distances.
They could be used to pinpoint its location given the positions of the walls, but the problem is harder and more interesting when we do not know where the walls are.
A similar principle is used by bats to echolocate, although we do not assume having any directional information.
% An appeal of a collocated setup is that it does not require any preinstalled infrastructure~\cite{Krekovic2016echoslam}.
Another problem that can be cast in this mold is the well-known ``structure from sound'' \cite{thrun2006affine}, where the task is to localize a set of microphones from phase differences induced by a set of unknown far field sources.

Prior work on localization from point-to-plane distances has so far been mostly computational \cite{kuang2012understanding, Krekovic2018structure}. Although several papers point out problems with uniqueness \cite{peng2015room, krekovic2017omnidirectional}, a complete study was up to now absent.
The most notable result is presented in~\cite{boutin2019drone}, which shows that one can reconstruct a room from the first-order echoes from one omnidirectional loudspeaker to four non-planar microphones, placed together on a drone with generic position and orientation.

In this work, we focus on uniqueness of reconstruction from point-to-plane distance matrices (PPDMs). Unlike in the case of localization from points, where with sufficiently many points the only possible ambiguity is that of translation, rotation, and reflection~\cite{dokmanic2015euclidean}, our analysis shows that localization from PPDMs suffers from additional ambiguities that correspond to certain continuous deformations of the points--planes system. 

\subsection{Related work}

The PPDM problem is related to the more standard multidimensional unfolding \cite{schonemann1970metric}: localization of a set of points from distances to a set of point landmarks. There are several variations of this problem that correspond to different assumptions of what is known: 1) given distances to known landmarks, localize unknown points (i.e., estimate the unknown trajectory), 2) given distances to known points, reconstruct unknown landmarks (i.e., map the unknown environment), 3) estimate both unknown landmarks and unknown points from their pairwise distances.

The first scenario is solved by simple multilateration \cite{bancroft1985algebraic}. The second scenario is a topic of active research in signal processing and room acoustics, where it is known as ``hearing the shape of a room''~\cite{Dokmanic2015thesis, Dokmanic2013, Antonacci2012}. Much of that work assumes that the geometry of the microphone array is known. If that is the case, since the source is fixed, the landmarks are modeled by points that correspond to virtual sources.

When neither the landmarks nor the points are known, we get an instance of SLAM. In general SLAM, the task is to simultaneously build some representation of the map of the environment and estimate the trajectory. Different flavors of SLAM involve different sensing modalities; prior work has considered visual~\cite{Davison2007, Clipp2010, Blosch2010, fuentes2015visual}, range-only~\cite{Djugash2006, Blanco2008, menegatti2009range}, and acoustic SLAM~\cite{Hu2011, zhou2017batmapper, evers2018acoustic, pradhan2018smartphone}, as well as solutions based on multiple sensor modalities~\cite{Brunner2011, kubelka2015robust, milford2016ratslam}.
Localization from PPDMs corresponds to range-only SLAM, though conventional approaches to SLAM rely on some noisy estimate of the trajectory, which is more information than we assume in scenario (3) above.

Methods for SLAM from reflections of sound or radio waves~\cite{Dokmanic2015thesis, Meissner2014, Leitinger2015, Dokmanic2014tc, evers2016acoustic} usually assume a fixed source or a fixed receiver, so that the echoes correspond to virtual beacons that provide range measurements. This information in turn allows to localize both sets using tools such as multidimensional unfolding~\cite{schonemann1970metric}. More recent works~\cite{Kietlinski2011thesis, Kietlinski2011, Meissner2014, leitingerbelief} show how to exploit multipath reflections. An appeal of our collocated setup is that it does not require any preinstalled infrastructure~\cite{Krekovic2016echoslam}.

\subsection{Our contributions}

We have previously shown that range-only SLAM can be addressed effectively using Euclidean distance matrices (EDM)~\cite{dokmanic2016acoustic}. Here we show how our new problem can be similarly cast as localization from PPDMs. This completes and extends our work on the 2D case \cite{Krekovic2016look}. Unlike in standard SLAM, we do not assume any motion model; the waypoints can be scattered arbitrarily.  

We study uniqueness of reconstruction of point--plane configurations from their pairwise distances. We derive conditions under which the localization is unique, and provide a complete characterization of non-uniqueness by enumerating the equivalence classes of configurations that lead to same PPDMs. Since we are motivated by SLAM, we refer to point--plane configurations as \textit{rooms} and \textit{trajectories}. The conclusions, however, are general, and can be applied to any of the discussed applications.

Finally, while PPDMs provide a good basic model for SLAM from echoes with a collocated source and receiver, the full SLAM problem presents a number of additional challenges. Problems of associating echoes to walls, dealing with missing echoes, and telling first-order from higher-order echoes will be addressed in a companion paper in preparation. Here we assume having a full PPDM as defined in Section \ref{problem_setup}.

\section{Problem setup}
\label{problem_setup}

Suppose that a mobile device carrying an omnidirectional source and an omnidirectional receiver traverses a trajectory described by $N$ waypoints $\{ {\vr_n} \}_{n=1}^N$. At every waypoint, the source produces a pulse, and the receiver registers the echoes.
Since the source and receiver are collocated, the distance $d_{nk}$ between the $n$th waypoint and $k$th wall is given by
\begin{align}
d_{nk} = \tfrac{1}{2} c \tau_{nk},
\end{align}
where $c$ is the speed of sound and $\tau_{nk}$ is the propagation time of the first-order echo. We can thus find the distances between waypoints and walls by measuring the times of arrival of first-order echoes.

To describe a room, we consider $K$ walls $\{ { {\cal P}_k } \}_{k=1}^K$ (lines in $2$D and planes in $3$D) defined by their unit normals $\vn_k \in \R^m$ and any point $\vp_k \in \R^m$ on the wall, where $m \in \set{2, 3}$. For any $\vx \in \mathcal{P}_k$ we have $\inprod {\vn_k, \vx} = q_k$, where $q_k = \inprod{\vn_k, \vp_k}$ is the distance of the wall from the origin.

Given the distances between walls and waypoints,
\begin{equation}
\label{eq:distdef}
d_{nk} = \mathrm{dist}(\vr_n, {\cal P}_k) = \inprod{\vp_k -\vr_{n},  \vn_{k}} = q_k - \inprod{\vr_{n},  \vn_{k}},
\end{equation}
for  $n = 1, ...,N$ and $k = 1, ..., K$, we define
\begin{equation}
\mD \bydef [d_{nk}]_{n,k=1}^{N, K} \in \mathbb{R}^ {N \times K}
\end{equation} to be the \textit{point-to-plane distance matrix} (PPDM); we always assume $N \geq K$.

By setting $\vq \bydef [ q_1, \hdots, q_K ]^\top$, $\mR \bydef [\vr_1, \hdots, \vr_N ]$, and $\mN \bydef [\vn_1, \hdots, \vn_K]$, we can express a PPDM as
\begin{align}
\label{eq:distdef2}
    \mD = \mathbf{1} \vq^\top - \mR^\top \mN,
\end{align}
where $\vq$ is the vector of distances between the planes and the origin, columns of $\mR \in \mathbb{R}^{m \times N}$ are the waypoint coordinates, and columns of $\mN \in \mathbb{R}^{m \times K}$ outward looking normal vectors of the planes.
Letting $\mP \bydef \begin{bmatrix} \vp_1 & \hdots & \vp_K \end{bmatrix}$, the vector $\vq$ can be written as $\vq = \diag(\mP^\top \mN)$, where $\diag(\mM)$ denotes the vector formed from the diagonal of $\mM$.

A pair of planes and waypoints defines a \textit{room--trajectory} configuration ${\cal R} = \left( \set {{\cal P}_k}_{k = 1}^K, \set {\vr_n}_{n = 1}^N \right)$, and the corresponding PPDM $\mD({\cal R})$.
In realistic convex configurations, all entries of the PPDM \eqref{eq:distdef2} are non-negative. However, in our relaxed definition of a room, the waypoints can lie on either side of a wall, so we allow for signed distances.

Our central question is whether a given PPDM $\mD({\cal R})$ specifies a unique room-trajectory configuration ${\cal R}$, or, equivalently, whether the map $\mathcal{R} \mapsto \mD(\mathcal{R})$ is injective. It is clear that rotated, translated, and reflected versions of $\mathcal{R}$ all give the same $\mD$, so we consider them to be the same configuration (we consider the equivalence class of all room--trajectory configurations modulo rigid motions and reflections).

We formalize the uniqueness question as follows:
\begin{problem}
\label{problem:one}
Are there distinct room--trajectory configurations ${\cal R}^1 = \left( \set {{\cal P}^1_k}_{k = 1}^K, \set {\vr_n^1}_{n = 1}^N \right)$ and ${\cal R}^2 = \left( \set {{\cal P}^2_k}_{k = 1}^K, \set {\vr^2_n}_{n = 1}^N \right)$ which are not rotated, translated, and reflected versions of each other, such that $\mD({\cal R}^1) = \mD({\cal R}^2)$?
\end{problem}

\section{Uniqueness of the reconstruction}
\label{sec:uniqueness}

Perhaps surprisingly, there are many examples of rooms from Problem \ref{problem:one}. The main tool in identifying the sought equivalence classes is the following lemma.

\begin{figure*}[t!]
\begin{align}
\label{eq:entries_N}
   \widebar{\mN}_{2D}^\top =
   \begingroup 
    \setlength\arraycolsep{1.22pt}
   \begin{bmatrix}
    \cos\varphi^0_1 & \sin\varphi^0_1 & \cos\varphi_1 & \sin\varphi_1 \\
    \cos\varphi^0_2 & \sin\varphi^0_2 & \cos\varphi_2 & \sin\varphi_2 \\
    \vdots & \vdots & \vdots & \vdots \\
    \cos\varphi^0_K & \sin\varphi^0_K & \cos\varphi_K & \sin\varphi_K
    \end{bmatrix},
    \widebar{\mN}_{3D}^\top =
    \begin{bmatrix}
    \sin\theta_1^0 \cos\varphi_1^0 & \sin\theta_1^0 \sin\varphi_1^0 & \cos\theta_1^0 & \sin\theta_1 \cos\varphi_1 & \sin\theta_1 \sin\varphi_1 & \cos\theta_1 \\
    \sin\theta_2^0 \cos\varphi_2^0 & \sin\theta_2^0 \sin\varphi_2^0 & \cos\theta_2^0 & \sin\theta_2 \cos\varphi_2 & \sin\theta_2 \sin\varphi_2 & \cos\theta_2 \\
    \vdots & \vdots & \vdots & \vdots & \vdots & \vdots  \\
    \sin\theta_K^0 \cos\varphi_K^0 & \sin\theta_K^0 \sin\varphi_K^0 & \cos\theta_K^0 &  \sin\theta_K \cos\varphi_K & \sin\theta_K \sin\varphi_K & \cos\theta_K
    \end{bmatrix}
    \endgroup
\end{align}
\noindent\hrulefill 
\end{figure*}

\begin{lemma} 
\label{lemma1}
Let ${\cal R}^0= \left( \set{{\cal P}_k^0}_{k=1}^K, \set{\vr^0_n}_{n=1}^N \right)$. Then for any ${\cal R}^1 = \left(\set{{\cal P}^1_k}_{k=1}^K, \set{\vr_n^1}_{n=1}^N \right)$ such that $\mD({\cal R}^0) = \mD({\cal R}^1)$, there exist a translation ${\cal R} = \left(\set{{\cal P}_k}_{k=1}^K, \set{\vr_n}_{n=1}^N \right)$ of ${\cal R}^1$ satisfying
\begin{align}
\label{eq:lemma_1}
\widebar{\mR}^\top \widebar{\mN} = \vec{0},
\end{align}
where
\begin{align}
\begin{aligned}
\label{eq:def_R_N}
&\widebar\mR \bydef \begin{bmatrix} \mR^0 \\ -\mR \end{bmatrix} =
\begin{bmatrix}
\vr^0_1  & \hdots & \vr^0_N  \\
-\vr_1  & \hdots & -\vr_N  \\
\end{bmatrix}, \\ 
&\widebar\mN \bydef \begin{bmatrix} \mN^0 \\ \mN \end{bmatrix} = \begin{bmatrix}  
\vn^0_{1}  & \hdots & \vn^0_{K} \\
\vn_{1}  & \hdots & \vn_{K}  
\end{bmatrix}.
\end{aligned}
\end{align}
Conversely, given $\mathcal{R}^0$, if \eqref{eq:lemma_1} holds for some $\mR$ and $\mN$, then there exists an ${\cal R} = \left(\set{{\cal P}_k}_{k=1}^K, \set{\vr_n}_{n=1}^N \right)$ with waypoints $\mR$ and wall normals $\mN$ such that $\mD({\cal R}^0) = \mD({\cal R})$.
\end{lemma}
\begin{proof}
We first prove the converse. Assume \eqref{eq:lemma_1} holds,
\begin{align*}
\label{eq:vector_form_lemma_proof}
\inprod{\vr^0_n, \vn^0_k} =\inprod{\vr_n, \vn_k},~\text{for all}~1 \leq n \leq N, \ 1 \leq k \leq K,
\end{align*}
and let
$u_k =\inprod{\vp^0_k, \vn^0_k} \mbox{ and } v_k = \inprod{\vp_k, \vn_k }$.
Note that we can always find $\vp_k$ such that $u_k = v_k$, for all $k \in \set{1, \ldots, K}$. For these $\vp_k$ we have
\begin{align}
- \inprod{\vr^0_n, \vn^0_k} + \inprod{\vp^0_k, \vn^0_k} = - \inprod{\vr_n, \vn_k} + \inprod{\vp_k, \vn_k}.
\end{align}
Letting $\calR$ be a configuration with waypoints $\mR$, wall normals $\mN$, and $k$th wall passing through $\vp_k$, the definition \eqref{eq:distdef} implies $\mD({\cal R}^0) = \mD({\cal R})$.

Now assume that for some $\mathcal{R}$ we have $\mD({\cal R}^0) = \mD({\cal R})$. Equivalently,
\begin{align}
\begin{aligned}
\label{eq:lemma_proof_eq}
\inprod{\vr^0_n, \vn^0_k} - \inprod{\vr_n, \vn_k} &=  \inprod{\vp^0_k, \vn^0_k} - \inprod{\vp_k, \vn_k} \\ &= q_k^0 - q_k,
\end{aligned}
\end{align}
for all $1 \leq n \leq N$, $1 \leq k \leq K$.

As we consider translated, rotated and reflected versions of $\mathcal{R}$ as the same $\mathcal{R}$, we can translate configurations $\mathcal{R}^0$ and $\mathcal{R}$ by $-\vr_1^0$ and $-\vr_1$, respectively. Hence, the waypoints for $n=1$ fall at the origin in both rooms, and~\eqref{eq:lemma_proof_eq} implies that $q_k^0 - q_k=0$ for $1 \leq k \leq K$, or $\widebar{\mR}^\top \widebar{\mN} = \vec{0}$ in the matrix form.
\end{proof}

We now characterize the pairs of configurations that satisfy~\eqref{eq:lemma_1}. In other words, we identify the equivalence classes of rooms and trajectories with respect to PPDMs.

The non-uniqueness condition~\eqref{eq:lemma_1} is satisfied when the columns of $\widebar{\mR}$ are in the nullspace of $\widebar{\mN}^\top$.
We parameterize the unit-norm columns of $\widebar{\mN}^\top = \begin{bmatrix} {\mN^0}^\top & \mN^\top \end{bmatrix}$ as
\begin{align}
    \vn^0_k = \begin{bmatrix}
    \cos\varphi^0_k\\
    \sin\varphi^0_k
    \end{bmatrix} \quad \mbox{and} \quad
    \vn_k = \begin{bmatrix}
    \cos\varphi_k\\
    \sin\varphi_k
    \end{bmatrix} 
\end{align}
in $2$D, and
\begin{equation}
\vn_k^0 =
\begin{bmatrix} \sin\theta_k^0 \cos\varphi_k^0 \\ \sin\theta_k^0 \sin\varphi_k^0 \\ \cos\theta_k^0
\end{bmatrix} \quad \mbox{and} \quad
\vn_k =
\begin{bmatrix} \sin\theta_k \cos\varphi_k\\ \sin\theta_k \sin\varphi_k \\ \cos\theta_k \end{bmatrix}
\end{equation}
in $3$D; $\widebar{\mN}^\top$ is written out in~\eqref{eq:entries_N}. 
The wall normals $\mN^0$ and $\mN$ of the two room-trajectory configurations $\mathcal{R}^0$ and $\mathcal{R}$ are uniquely determined by the angles $\left\{\varphi^0_k \right\}_{k=1}^{K}$ and $\left\{ \varphi_k \right\}_{k=1}^{K}$ in $2$D, or by the pairs of angles $\left\{\theta^0_k, \varphi^0_k\right\}_{k=1}^{K}$ and $\left\{\theta_k, \varphi_k \right\}_{k=1}^{K}$ in $3$D, where $\varphi^0_k, \varphi_k \in [0, 2\pi)$ and $\theta^0_k, \theta_k \in [0, \pi)$. As the converse is also true---the matrix $\widebar{\mN}$ uniquely determines the angles---we interchangeably use both notations.

For $K \geq 2m$, the nullspace of $\widebar{\mN}^\top$ is generically empty. To find the configurations that are not uniquely determined by PPDMs, we impose linear dependencies among its columns: we select any $r$ linearly independent columns of $\widebar{\mN}^\top$ and assume that the remaining columns are their linear combinations.
Restricting the analysis to a particular column selection does not reduce generality, as shown in 
Appendix.

In addition to these linear dependencies, the columns in~\eqref{eq:entries_N} are also subject to non-linear relationships due to the normalization constraint. Indeed, $\widebar{\mN}^\top$ has $K$ rows, $2m$ columns, and only $2(m-1) K$ free parameters.
The combination of these linear and non-linear dependencies determines the equivalence classes of the rooms and trajectories with respect to PPDMs. Our goal is to characterize these classes. 

Specifically, for every equivalence class we want to find a reference configuration $\mathcal{R}^0$ that identifies the class, and a rule that generates other $\mathcal{R}$ with the same PPDM.
Letting $r = \rank(\widebar{\mN})$, the analysis is performed for every $r \in \set{1, \ldots, 2m - 1}$ in six steps. We introduce and explain those steps on the case $r=2$ in $2$D, rather than $r=1$ which gives degenerate solutions (we analyze $r=1$ subsequently).

As we will see, most of the identified cases correspond to rooms that are in some sense degenerate (for example, a ``room'' with all walls parallel), although as point--plane configurations they are perfectly reasonable.

The analysis in Section \ref{sec:classification_2d} and Section \ref{sec:classification_3d} together with the fact that  Lemma \ref{lemma1} is sufficient and necessary prove that the union of all equivalence classes described in this paper (see Fig.~\ref{fig:classification}) is in fact the set of all possible configurations that are not uniquely determined by their PPDM. In other words, a room can be uniquely reconstructed from a PPDM (modulo rigid motions) if and only if it does not belong to one of the classes illustrated in Fig.~\ref{fig:classification}.

\begin{theorem}
In 2D, a room--trajectory configuration is not uniquely determined by its PPDM if and only if at least one of the following holds: 1) waypoints are collinear, 2) all walls are parallel (infinitely long corridors), 3) walls form a parallelogram possibly extended by parallel walls (see Fig.~\ref{fig:classification}).

In 3D, a room--trajectory configuration is not uniquely determined by its PPDM if and only if at least one of the following holds: 1) $K < 6$, 2) waypoints are coplanar, 3) the configuration is in one of the classes summarized in Fig.~\ref{fig:classification}.
\label{thm:main}
\end{theorem}

\begin{figure*}[t]
    \centering
    \includegraphics[width=\linewidth]{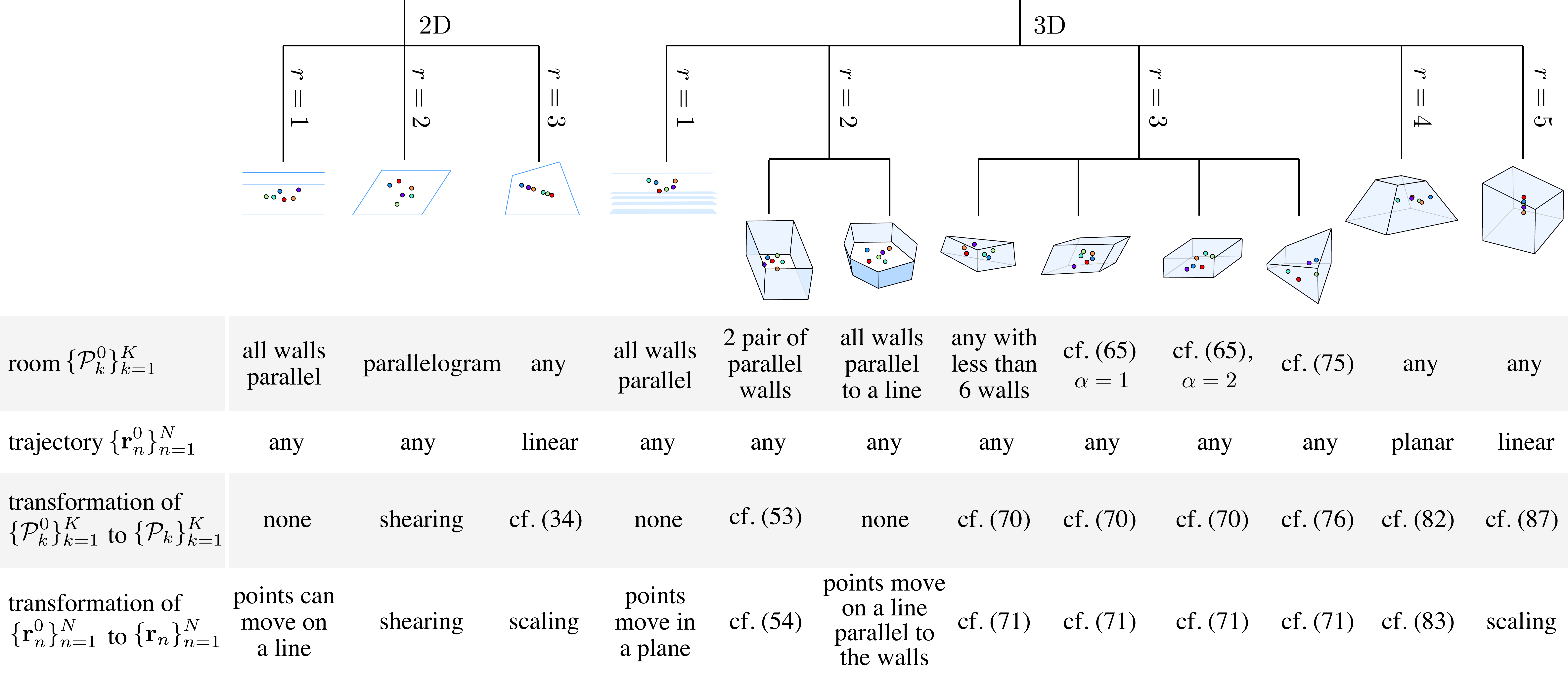}
    \caption{An overview of the equivalence classes of the room-trajectory configurations with respect to PPDMs in $2$D and $3$D. The class generator is denoted by ${\cal R}^0= \left( \set{{\cal P}_k^0}_{k=1}^K, \set{\vr^0_n}_{n=1}^N \right)$ and the equivalent configurations by ${\cal R}= \left( \set{{\cal P}_k}_{k=1}^K, \set{\vr_n}_{n=1}^N \right)$.}
    \label{fig:classification}
\end{figure*}

% ------------------------------------------------
\section{Classification of 2D configurations}
\label{sec:classification_2d}
We begin by the easier $2$D analysis, i.e. $m=2$. For $\widebar{\mN}^\top$ to have a nullspace, we must have $r \in \set{1, 2, 3}$. For all $r$ the analysis is performed as a sequence of six steps, which we describe in detail for $r=2$.

% ------------------------------------------------
\subsection{2D rank-2: Parallelogram rooms}

\begin{enumerate}[leftmargin=*, wide, labelwidth=!, labelindent=0pt]

\item \textit{Linear dependence:}
We select $r$ linearly independent columns of $\widebar{\mN}^\top$, denoted $\vc_i \in \mathbb{R}^K$, $i = 1, \hdots, r$, and denote the remaining columns of $\widebar{\mN}^\top$ by $\vc_k \in \mathbb{R}^K$, $k = r+1, \hdots, 2m$. We let $\vc_k$ for $k > r$ be linear combinations of $\vc_k$ for $k \leq r$: 
\begin{align}
\label{eq:system}
    \begin{bmatrix} \vc_{r+1} & \hdots & \vc_{2m} \end{bmatrix}^\top = \mT \begin{bmatrix} \vc_1 & \hdots & \vc_r \end{bmatrix}^\top,
\end{align}
for some $\mT \in \mathbb{R}^{(2m-r) \times r}$.

Concretely, for $r=2$, we assume that the first two columns of $\widebar{\mN}^\top$ are linearly independent, while the third and the fourth column are their linear combinations. We prove in Appendix that this particular choice of columns does not incur a loss of generality in this or any of the other cases. From \eqref{eq:entries_N}, for every $k$ we have that
\begin{align}
\begin{bmatrix}
\cos{\varphi_k}  \\
\sin{\varphi_k} 
\end{bmatrix} = 
\mT
\begin{bmatrix}
\cos{\varphi^0_k}  \\
\sin{\varphi^0_k} 
\end{bmatrix}, \mbox{ where } \mT = \begin{bmatrix}
a & b\\
c & d
\end{bmatrix}.
\label{eq:2d_rank2_definition}
\end{align}

\item \textit{Reparametrization:} 
When $r \leq m$, we can rearrange the columns so that the right-hand side of~\eqref{eq:system} contains the normals of the reference configuration $\mathcal{R}^0$, while the left-hand side has the normals of the putative equivalent configuration $\mathcal{R}$. In particular, we obtain
\begin{align}
\label{eq:subsystem}
    {\mN}^\top = \mT' {\mN^0}^\top,
\end{align}
where $\mT' \in \mathbb{R}^{m \times m}$. 
$\mT'$ can be decomposed as a product $\mT' =\mQ \mU$ of an orthogonal matrix $\mQ$ and an upper triangular matrix $\mU$.
$\mQ$ acts as a rotation and a reflection, so without loss of generality we set $\mQ = \mI$ and $\mT' = \mU$. That is, we assume that the entries of $\mT'$ below the diagonal are $0$, which removes the rotational degrees of freedom.
Since~\eqref{eq:subsystem} contains a subset of equations from~\eqref{eq:system}, we propagate this change back to~\eqref{eq:system} by modifying the corresponding elements of $\mT$.

When $r=m$, the original system of equations~\eqref{eq:2d_rank2_definition} already has a form of~\eqref{eq:subsystem}. Therefore, we only need to set $c = 0$ and obtain an upper triangular matrix,
\begin{align}
\mT = 
\begin{bmatrix}
a & b \\
0 & d \\
\end{bmatrix}.
\end{align}

\item \textit{Reference room:}
To find a reference room, we select an arbitrary $\mT$ (respecting the zero entries from step 2), and solve for the normals satisfying~\eqref{eq:system}. From~\eqref{eq:2d_rank2_definition}, we observe that
\begin{equation}
\label{eq:2d_rank2_squared_sum}
    (a \cos \varphi_k^0 + b \sin \varphi_k^0)^2 + (d \sin \varphi_k^0) ^2 = \cos \varphi_k ^2 + \sin \varphi_k^2 = 1,
\end{equation}
so the angles of the reference room cannot be chosen arbitrarily.
To find the values of $\{\varphi_k^0 \}_{k=1}^K$ with respect to free parameters $a$, $b$ and $d$, we solve~\eqref{eq:2d_rank2_squared_sum} and obtain
\begin{align}
A \cos^2{(2\varphi^0_k)} + B \cos{(2\varphi^0_k)} + C = 0, 
\label{eq:2d_rank2_ABC}
\end{align} 
where 
\begin{equation}
\begin{aligned}
A &= (a^2 - b^2 - d^2)^2 + 4 a^2 b^2,\\
B &= 2 (a^2 - b^2 - d^2)(a^2 + b^2 + d^2 - 2), \\
C &= (a^2 + b^2 + d^2 - 2)^2 - 4 a^2 b^2.
\end{aligned}
\end{equation}

Let first $A = 0$. Then \eqref{eq:2d_rank2_ABC} has two solutions: $a = 0$, $b^2 = -d^2$ and $b = 0$, $a^2 = d^2$. The first one implies that $b = d = 0$, which makes~\eqref{eq:2d_rank2_definition} inconsistent. The second one leads to $\mT$ being a reflection matrix:
\begin{align}
\mT =
\begin{bmatrix}
\pm 1 & 0\\
0 & \pm 1
\end{bmatrix} 
\quad \mbox{or} \quad
\mT =
\begin{bmatrix}
\pm 1 & 0 \\
0 & \mp 1
\end{bmatrix},
\end{align}
which is not of our interest.  

For $A \neq 0$, we have
\begin{align}
\label{eq:2d_rank2_quadratic0}
    \cos {(2 \varphi^0_k)} = \frac{-B \pm \sqrt{B^2 - 4AC}}{2A}.
\end{align}
There are eight solutions for $\varphi^0_k$, four of which satisfy~\eqref{eq:2d_rank2_definition}. 
The valid solutions always come as pairs $(\varphi_1^0, \varphi_2^0) = (\varphi_1^0, \varphi_1^0 + \pi)$ and $(\varphi_3^0, \varphi_4^0) = (\varphi_3^0, \varphi_3^0 + \pi)$.

\item \textit{Equivalent rooms:}
From~\eqref{eq:system}, we identify the transformation that takes the normals of the reference room to the normals of an equivalent room. The corresponding angles in the equivalent room are computed from~\eqref{eq:2d_rank2_definition},
\begin{equation}
    \varphi_k = f(\varphi_k^0, s_k, a, b, d) = \arctan \frac{d \sin \varphi_k^0}{a \cos \varphi_k^0 + b \sin \varphi_k^0} + s_k \pi,
\end{equation}
where $s_k \in \{0, 1\}$.

\item \textit{Equivalence class:} The solutions of~\eqref{eq:2d_rank2_quadratic0} suggest that we can construct a reference room by arbitrarily choosing two wall normals, $\varphi^0_1$ and $\varphi^0_3$, and solving the system of two equations~\eqref{eq:2d_rank2_ABC} with $k \in \{1, 3\}$. This fixes two parameters (e.g., $a$ and $b$) in $\mT$ and leaves the third (e.g., $d$) free to generate an infinite number of rooms equivalent to the reference room. Reference rooms are not restricted to only two walls; we can have any number of additional walls parallel to those determined by $\varphi^0_1$ and $\varphi^0_3$, since they also satisfy~\eqref{eq:2d_rank2_quadratic0}. 

A room with walls $\{\mathcal{P}^0_k\}_{k=1}^K =\{ (\vn^0_k, q^0_k) \}_{k=1}^K$ is a generator of a class of rooms with identical PPDMs, with normals $\vn^0_k$ chosen as described above
and $\vq^0 \in \mathbb{R}^K$. Using $\big[ \cdot \big]$ to denote equivalence classes, the above analysis defines the following equivalence class of rooms with the same PPDMs:
\begin{align}
    \Big[ \{\mathcal{P}^0_k\}_{k=1}^K \Big] = \Big\{  \{\mathcal{P}_k\}_{k=1}^K \, \big| \, \varphi_{k} = f(\varphi_{k}^0, s_k, a, b, c),& \nonumber \\ \, a,b \in \mathbb{R} \mbox { s.t.}~\eqref{eq:2d_rank2_ABC} \mbox{ satisfied}, d \in \mathbb{R}, s_k \in \{0,1\},& \nonumber \\ \, q_k = q^0_k \text{ for } 1 \leq k \leq K &\Big\}.
\end{align}

There are no constraints on the distances of walls from the origin in the reference room and we can set $\vq^0$ arbitrarily. The equivalent room satisfies $\vq = \vq^0$ by Lemma~\ref{lemma1}. We note that this class includes parallelogram rooms for $K=4$, $\varphi^0_2 = \varphi^0_1 + \pi$ and $\varphi^0_4 = \varphi^0_3 + \pi$.

\item \textit{Corresponding trajectories:}
Finally, we find the waypoints $\{\vr_n^0 \}_{n=1}^N$ and $\{ \vr_n \}_{n=1}^N$ that lie in the nullspace of $\widebar{\mN}^\top$.
The nullspace basis can be found as:
\begin{align*}
&\vv_1 =
\begin{bmatrix}
 -a, & -b, & 1, & 0
\end{bmatrix}^\top, \\
&\vv_2 =
\begin{bmatrix}
 0, & -d, & 0, & 1
\end{bmatrix}^\top,
\end{align*}
so the columns of $\widebar{\mR}$ are of the form
\begin{align}
\begin{bmatrix}
\vr_n^0 \\ -\vr_n
\end{bmatrix}=
\vv_1 \gamma_1 +
\vv_2 \gamma_2 
\label{eq:general},
\end{align}
where $\gamma_1, \gamma_2 \in \mathbb{R}$.
The waypoints in the reference room are chosen without restrictions, while the waypoints in the equivalent room are obtained by applying a non-rigid transformation
\begin{align}
    \vr_n^0 = \mT^\top \vr_n.
\end{align}

\end{enumerate}
An example of three parallelogram configurations with the same PPDM is illustrated in Fig.~\ref{fig:case_3}.
\begin{figure}[H]
    \centering
    \includegraphics[width=\linewidth]{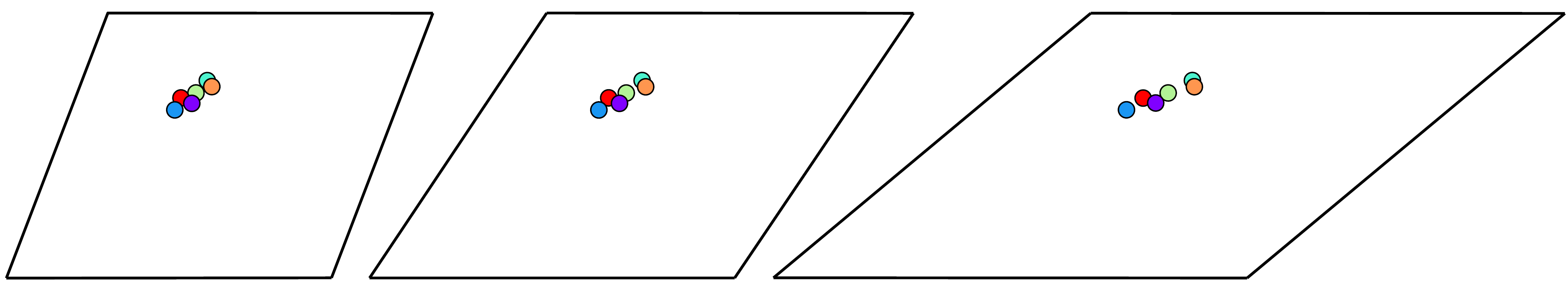}
    \caption{Parallelogram rooms with the same PPDM.}
    \label{fig:case_3}
\end{figure}

% ------------------------------------------------

\subsection{2D rank-1: Infinitely long corridors}

\begin{enumerate}[leftmargin=*, wide, labelwidth=!, labelindent=0pt]

\item \textit{Linear dependence:}
In $2$D, setting $\rank(\widebar{\mN}) = 1$ leads to degenerate rooms. To show that, assume that every column of $\widebar{\mN}^\top$ is a scaled version of the first column,
\begin{align}
\label{eq:2d_rank_1}
\begin{bmatrix}
\sin \varphi_k^0  \\ \cos \varphi_k \\ \sin \varphi_k
\end{bmatrix} = \mT
\cos \varphi_k^0, \mbox{ where } \mT = \begin{bmatrix}
a \\ b \\ c
\end{bmatrix}.
\end{align}

\item \textit{Reparametrization:}
These dependencies can be partially expressed as a transformation of the normals of the reference room to those of the equivalent room. From \eqref{eq:2d_rank_1} we have:
\begin{align}
\begin{bmatrix}
\cos{\varphi_k}  \\
\sin{\varphi_k} 
\end{bmatrix} = 
\mT'
\begin{bmatrix}
\cos{\varphi^0_k}  \\
\sin{\varphi^0_k}
\end{bmatrix}, 
\mbox{ where }
\mT' = \begin{bmatrix}
b & 0 \\
c & 0
\end{bmatrix}.
\end{align}
With $c=0$, $\mT'$ becomes upper triangular. This eliminates rotations and reflections.
Propagating back to $\mT$, we get:
\begin{align}
\label{eq:2d_rank_1_newT}
\mT = \begin{bmatrix}
a, \ b, \ 0
\end{bmatrix}^\top.
\end{align}

\item \textit{Reference room:}
We see that~\eqref{eq:2d_rank_1} constrains the normals of the reference room, since
\begin{align}
\label{eq:2d_rank1_sol}
\tan \varphi_k^0 = a
\end{align}
must hold for every $k$. That is, the wall normals of the reference room cannot be chosen arbitrarily. Letting $s_k \in \{ 0, 1\}$, we summarize both solutions to~\eqref{eq:2d_rank1_sol} as
\begin{align}
\label{eq:2d_rank1_definition}
 \varphi_k^0 = f(s_k, a) =  \arctan a + s_k \pi.
\end{align}
For $K \geq 2$ walls,~\eqref{eq:2d_rank1_definition} implies that every $\varphi_k^0$ can only assume two values. These correspond to parallel walls since $\varphi^0_i = \varphi^0_k + \pi$ for $s_i = 0$ and $s_k=1$. 

\item \textit{Equivalent rooms:}
From~\eqref{eq:2d_rank_1} and~\eqref{eq:2d_rank_1_newT} we have $\varphi_k \in \set{0, \pi}$ and $a^2 + 1 = b^2$.

\item \textit{Equivalence class:}
This trivial case results in the equivalence class of rooms with parallel walls. They are generated by a reference room $\{\mathcal{P}^0_k\}_{k=1}^K$ with the wall normals from~\eqref{eq:2d_rank1_definition} and $\vq^0 \in \mathbb{R}^K$,
\begin{align}
\label{eq:2d_rank1_ec}
 \Big[ \{\mathcal{P}^0_k\}_{k=1}^K \Big] = \Big\{  \{\mathcal{P}_k\}_{k=1}^K \, \big| \, \varphi_k \in \set{0, \pi}, q_k = q^0_k,&\nonumber \\ \, \text{for } 1 \leq k \leq K &\Big\}.
\end{align}

\item \textit{Corresponding trajectories:}
Though all rooms in this class have the same geometry, there are infinitely many trajectories that lead to the same PPDM. To see this, imagine an infinite corridor with two parallel walls. The points on any line parallel to the walls cannot be discriminated from distances to walls.
Formally, a basis for the nullspace of $\widebar{\mN}^\top$ is 
\begin{align*}
    \vv_1 = \begin{bmatrix}
    -a \\ 1 \\ 0 \\ 0
    \end{bmatrix}, 
    \vv_2 = \begin{bmatrix}
    -b \\ 0 \\ 1 \\ 0
    \end{bmatrix},
    \vv_3 = \begin{bmatrix}
    0 \\ 0 \\ 0 \\ 1
    \end{bmatrix},
\end{align*}
so the columns of $\widebar{\mR}$ have to be of the form
\begin{align}
\label{eq:2d_rank_1_waypoints}
    \begin{bmatrix}
    \vr_n^0 \\
    -\vr_n
    \end{bmatrix} = \gamma_1 \vv_1 + \gamma_2 \vv_2 + \gamma_3 \vv_3,
\end{align}
where $\gamma_1, \gamma_2$ and $\gamma_3 \in \mathbb{R}$. This further implies that the waypoints of the reference room $\left\{ \vr^0_n \right\}_{n=1}^N$ and the $y$ coordinates of $\left\{ \vr_n \right\}_{n=1}^N$ in the equivalent rooms are independent and the latter can be chosen arbitrarily. The $x$ coordinates of $\left\{ \vr_n \right\}_{n=1}^N$ are given by~\eqref{eq:2d_rank_1_waypoints}. Fig.~\ref{fig:case_1} shows three equivalent configurations that emerge from this case.
\end{enumerate}

\begin{figure}[H]
    \centering
    \includegraphics[width=\linewidth]{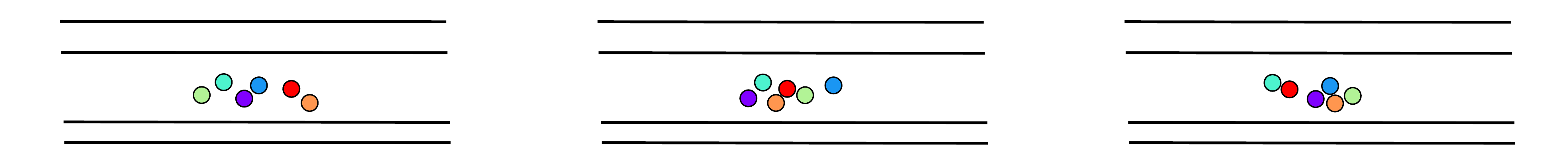}
    \caption{Example of three equivalent infinitely long corridors.}
    \label{fig:case_1}
\end{figure}

% ------------------------------------------------
\subsection{2D rank-3: Linear trajectories}

\begin{enumerate}[leftmargin=*, wide, labelwidth=!, labelindent=0pt]

\item \textit{Linear dependence:}
We assume $\rank(\widebar{\mN}) = 3$ so that
\begin{align}
\label{eq:2d_rank_3}
\cos{\varphi_k}  = 
\mT \begin{bmatrix}
\sin{\varphi_k}  \\
\cos{\varphi^0_k} \\
\sin{\varphi^0_k}
\end{bmatrix}, \mbox{ where } \mT = \begin{bmatrix}
a & b & c\\
\end{bmatrix}.
\end{align}

\item \textit{Reparametrization:}
As $r > m$, we cannot rewrite~\eqref{eq:2d_rank_3} such that the wall normals of $\mathcal{R}$ and $\mathcal{R}^0$ are on the opposite sides of the equation, so we omit this step.

\item \textit{Reference room:}
From~\eqref{eq:2d_rank_3}, we observe that the wall orientations of the reference room are unconstrained.

\item \textit{Equivalent rooms:}
We can express the wall orientations $\varphi_k$ in the equivalent room as a function of $\varphi_k^0$ and entries in $\mT$,
\begin{equation}
\varphi_k = f(\varphi_k^0, s_k, a, b, c) = s_k \arccos  \frac{b \cos \varphi_k^0 + c \sin \varphi_k^0}{\sqrt{a^2+1}} - \arctan\!a,
\end{equation}
where $s_k \in \{ -1, 1\}$.

\item \textit{Equivalence class:}
An arbitrary room $\{\mathcal{P}^0_k\}_{k=1}^K$ with $K$ walls generates the following equivalence class:
\begin{align}
    \Big[ \{\mathcal{P}^0_k\}_{k=1}^K \Big] = \Big\{ \{\mathcal{P}_k\}_{k=1}^K \big| & \; \varphi_k = f(\varphi_k^0, s_k,a,b, c), \nonumber \\ & \; a,b,c \in \mathbb{R}, s_k \in \{-1,1\}, \nonumber \\ & \; q_k = q^0_k, \text{ for } 1 \leq k \leq K  \Big\}.
\end{align}

\item \textit{Corresponding trajectories:}
The nullspace of $\widebar{\mN}^\top$ is spanned by one vector
$\vv_1 = \begin{bmatrix} -b & -c & 1 & -a \end{bmatrix}^\top$, so the columns of $\widebar{\mR}$ satisfy
\begin{equation}
\begin{bmatrix} \vr_n^0 \\ -\vr_n \end{bmatrix} = \vv_1 \gamma,
\end{equation}
where $\gamma \in \mathbb{R}$. This further suggests that the $x$ and $y$ coordinates of the waypoints in both rooms are dependent, and the trajectories are linear. 

\end{enumerate}

In other words, for any arbitrary room with $K$ walls and a PPDM measured at collinear waypoints, we can find another room with the same PPDM obtained at different collinear waypoints; an example is shown in Fig.~\ref{fig:case_2}.

\begin{figure}[H]
    \centering
    \includegraphics[width=\linewidth]{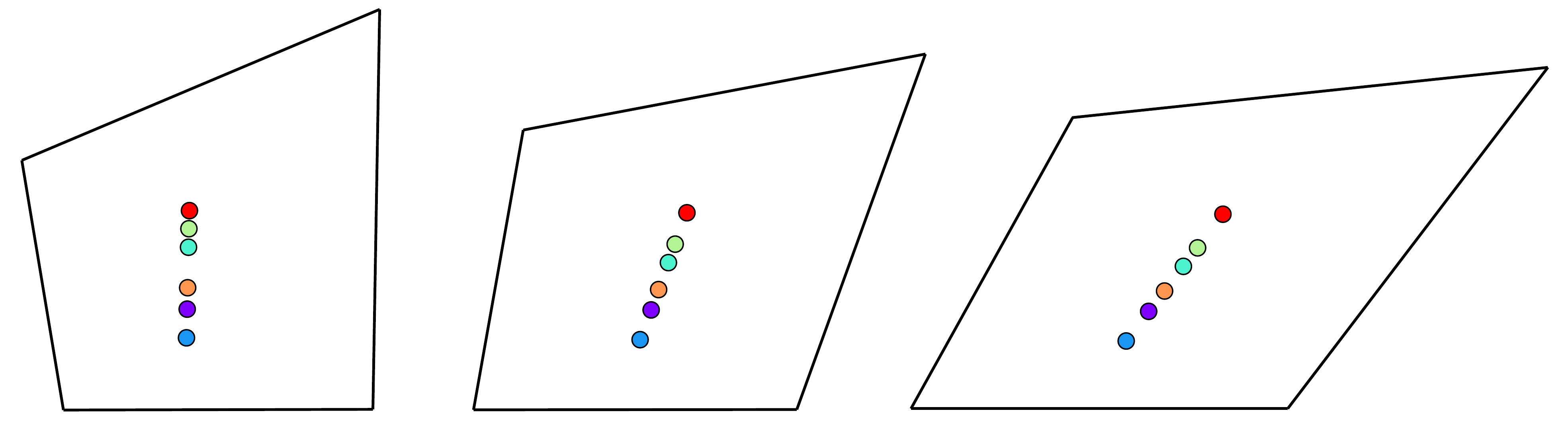}
    \caption{Example of equivalent rooms with linear trajectories.}
    \label{fig:case_2}
\end{figure}

% ------------------------------------------------

\section{Classification of 3D configurations}
\label{sec:classification_3d}
In $3$D $(m = 3)$ we analyze the cases $r\in \{1, 2, 3, 4, 5\}$.

\subsection{3D rank-1: Infinitely long and tall corridors}

\begin{enumerate}[leftmargin=*, wide, labelwidth=!, labelindent=0pt]

\item \textit{Linear dependence:}
When $\rank(\widebar{\mN}) = 1$ in $3$D, five columns of $\widebar{\mN}^\top$ are scaled version of a single non-zero column,
\begin{align}
\label{eq:3d_rank1_a}
\begin{bmatrix}
\sin\theta_k^0 \sin\varphi^0_k \\
\cos\theta_k^0 \\
\sin\theta_k \cos\varphi_k \\
\sin\theta_k \sin\varphi_k \\
\cos\theta_k \\
\end{bmatrix} =
\mT \sin\theta_k^0 \cos\varphi^0_k, \text{where }
\mT = \begin{bmatrix}
a \\ b \\ c \\ d \\ e
\end{bmatrix}.
\end{align}

\item \textit{Reparametrization:}
The requirement \eqref{eq:3d_rank1_a} implies the following relationship between the wall normals of the reference room and those of the equivalent room:
\begin{align}
\begin{aligned}
\label{eq:3d_rank1_definition}
\begin{bmatrix}
\sin\theta_k \cos\varphi_k  \\ \sin\theta_k \sin\varphi_k \\ \cos\theta_k
\end{bmatrix} =
\begin{bmatrix}
c & 0 & 0 \\
d & 0 & 0 \\
e & 0 & 0
\end{bmatrix}
\begin{bmatrix}
\sin\theta^0_k \cos\varphi^0_k  \\ \sin\theta^0_k \sin\varphi^0_k \\ \cos\theta^0_k
\end{bmatrix},
\end{aligned}
\end{align}
As before, we set $d = e = 0$ to get an upper triangular matrix.

\item \textit{Reference room:}
From~\eqref{eq:3d_rank1_a}, it follows that
\begin{align}
\tan \varphi_k^0 = a \quad \mbox{and} \quad
\tan \theta_k^0 = \frac{1}{b \cos \varphi_k^0}
\end{align} 
for every $k$. 
Then,
\begin{equation}
    \label{eq:3d_rank1_angles}
    \varphi_k^0 = \arctan a + s_k \pi, \quad 
    \theta_k^0 =  \arctan \frac{1}{b \cos \varphi_k^0} + t_k \pi,
\end{equation}
where $s_k$, $t_k \in \{0, 1\}$ are independent binary variables. That is, the reference room cannot be chosen arbitrarily; the angles can only assume two values that yield parallel walls.

\item \textit{Equivalent rooms:}
From~\eqref{eq:3d_rank1_a} we also find that $\sin \varphi_k = 0$ and $\cos \theta_k = 0$, so the angle $\theta_k$ takes a value of $\pi/2$, while $\varphi_k$ is either $0$ or $\pi$.
Furthermore, for~\eqref{eq:3d_rank1_a} to be consistent, 
\begin{align}
    c = \frac{\sin \theta_k \cos \varphi_k}{\sin \theta_k^0 \cos \varphi_k^0}.
\end{align}

\item \textit{Equivalence class:}
An equivalence class of these degenerate rooms with parallel walls with respect to PPDM is generated by a reference room $\{\mathcal{P}^0_k\}_{k=1}^K$ with the wall normals from~\eqref{eq:3d_rank1_angles} and $\vq^0 \in \mathbb{R}^K$,
\begin{align}
    \Big[ \{\mathcal{P}^0_k\}_{k=1}^K \Big] = \Big\{  \{\mathcal{P}_k\}_{k=1}^K \, \big| &\, \theta_k = \pi/2, \; \varphi_k \in \{0, \pi\}, \nonumber \\ &\, q_k = q^0_k, \text{ for } 1 \leq k \leq K \Big\}.
\end{align}

\item \textit{Corresponding trajectories:}
Analogously to the $\rank$-$1$ case in $2$D, the ambiguity in the reconstruction is due to the multitude of consistent trajectories. Points in planes parallel to the walls cannot be uniquely determined from distances to the walls.
The nullspace of $\widebar{\mN}^\top$ is spanned by five vectors,
\begin{align*}
\resizebox{\hsize}{!}{%
    $\vv_1 = \begin{bmatrix}
    -a \\ 1 \\ 0 \\ 0 \\ 0 \\ 0
    \end{bmatrix}, 
    \vv_2 = \begin{bmatrix}
    -b \\ 0 \\ 1 \\ 0 \\ 0 \\ 0
    \end{bmatrix},
    \vv_3 = \begin{bmatrix}
    -c \\ 0 \\ 0 \\ 1 \\ 0 \\ 0
    \end{bmatrix},
    \vv_4 = \begin{bmatrix}
    0 \\ 0 \\ 0 \\ 0 \\ 1 \\ 0
    \end{bmatrix},
    \vv_5 = \begin{bmatrix}
    0 \\ 0 \\ 0 \\ 0 \\ 0 \\ 1
    \end{bmatrix}, $}
\end{align*} 
so the columns of $\widebar{\mR}$ are
\begin{align}
\label{eq:3d_rank_1_waypoints}
    \begin{bmatrix}
    \vr_n^0 \\
    -\vr_n
    \end{bmatrix} = \gamma_1 \vv_1 + \gamma_2 \vv_2 + \gamma_3 \vv_3 + \gamma_4 \vv_4 + \gamma_5 \vv_5,
\end{align}
where $\gamma_1, \gamma_2, \gamma_3, \gamma_4$ and $\gamma_5 \in \mathbb{R}$.
This implies that the waypoints $\left\{ \vr^0_n \right\}_{n=1}^N$ in the reference room and the $y$ and $z$ coordinates of $\left\{ \vr_n \right\}_{n=1}^N$ in the equivalent rooms are independent and can be chosen arbitrarily, whereas the $x$ coordinates of $\left\{ \vr_n \right\}_{n=1}^N$ are given by~\eqref{eq:3d_rank_1_waypoints}. An example of such room-trajectory configurations is shown in Fig.~\ref{fig:case_1_3D}.
\end{enumerate}

\begin{figure}[H]
    \centering
    \includegraphics[width=\linewidth]{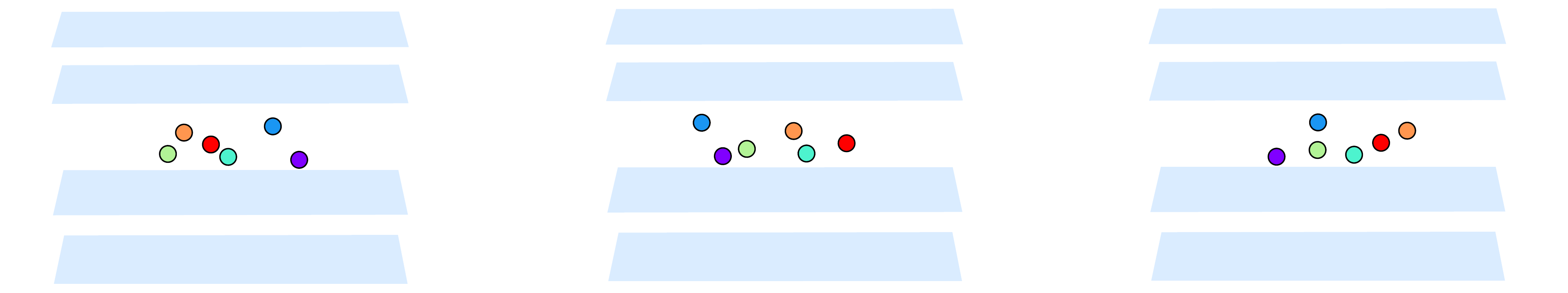}
    \caption{Three equivalent infinitely long and tall corridors.}
    \label{fig:case_1_3D}
\end{figure}

% ------------------------------------------------
\subsection{3D rank-2: Parallelepipeds without bases}
\label{sec:3drank2-I}

\begin{enumerate}[leftmargin=*, wide, labelwidth=!, labelindent=0pt]

\item \textit{Linear dependence:}
Assume that the first and the second column are linearly independent, and the others are their linear combinations. Thus, for every wall $k$ we have
\begin{align}
\label{eq:3d_rank2_definition}
\begin{bmatrix}
\cos\theta^0_k  \\
\sin\theta_k \cos\varphi_k \\
\sin\theta_k \sin\varphi_k \\
\cos\theta_k 
\end{bmatrix} = 
\mT 
\begin{bmatrix}
\sin\theta^0_k \cos\varphi^0_k \\
\sin\theta^0_k \sin\varphi^0_k 
\end{bmatrix}, \text{ where }
\mT = 
\begin{bmatrix}
a & b \\
c & d \\
e & f \\
g & h
\end{bmatrix}.
\end{align}

\item \textit{Reparametrization:}
As before, \eqref{eq:3d_rank2_definition} implies a relationship between the normals of the reference and the equivalent room,
\begin{align}
\label{eq:3d_rank2_reparametrization}
\begin{bmatrix}
\sin\theta_k \cos\varphi_k \\
\sin\theta_k \sin\varphi_k \\
\cos\theta_k 
\end{bmatrix} = 
\begin{bmatrix}
c & d & 0 \\
e & f & 0 \\
g & h & 0
\end{bmatrix}
\begin{bmatrix}
\sin\theta^0_k \cos\varphi^0_k \\
\sin\theta^0_k \sin\varphi^0_k  \\
\cos \theta^0_k
\end{bmatrix}.
\end{align}
By setting $e, g$ and $h$ to $0$, we obtain the desired upper triangular matrix and propagate this change into $\mT$,
\begin{align}
\mT = 
\begin{bmatrix}
a & c & 0 & 0 \\
b & d & f & 0 
\end{bmatrix}^\top.
\end{align}

\item \textit{Reference room:}
The sum of the squares of the last three equations in~\eqref{eq:3d_rank2_definition} has to be $1$ for every wall $k$,
\begin{equation}
   (c \sin \theta_k^0 \cos \varphi_k^0 + d \sin \theta_k^0 \sin \varphi_k^0)^2 + (f \sin \theta_k^0 \sin \varphi_k^0)^2 = 1,
\label{eq:3d_squared_sum}
\end{equation}
so the reference room cannot be chosen arbitrarily.
From~\eqref{eq:3d_squared_sum}, we can express $\theta_k^0$ as a function of $\varphi_k^0$ and the entries of $\mT$.

The first equation in~\eqref{eq:3d_rank2_definition} additionally constrains $\theta_k^0$ and $\varphi_k^0$,
\begin{align}
\label{eq:3d_rank2_first}
    \tan \theta_k^0 = (a \cos \varphi_k^0 + b \sin \varphi_k^0)^{-1}.
\end{align}
We obtain a quadratic equation with respect to $\cos (2 \varphi_k^0)$,
\begin{equation}
\label{eq:3d_ABC}
    (A^2 + B^2) \cos^2 (2 \varphi_k^0) - 2 A C \cos (2 \varphi_k^0) + (C^2 - B^2) = 0,
\end{equation}
where
\begin{align}
\begin{aligned}
\label{eq:3d_rank2_definitionsABC}
    A &= -a^2 + b^2 + c^2 - d^2 - f^2, \\
    B &= 2 (a b - c d), \\
    C &= a^2 + b^2 - c^2 - d^2 - f^2 + 2.
\end{aligned}
\end{align}

We first assume $A^2 + B^2 \neq 0$ and solve~\eqref{eq:3d_ABC} for $\varphi^0_k$,
\begin{align}
\label{eq:3d_rank2_quadratic0}
    \cos {(2 \varphi^0_k)} = \frac{AC \pm \sqrt{A^2 C^2 - (A^2 + B^2)(C^2 - B^2)}}{A^2+B^2}.
\end{align}
We obtain four solutions for $\varphi_k^0$ to~\eqref{eq:3d_ABC} that satisfy~\eqref{eq:3d_rank2_definition}.
For each value of $\varphi_k^0$ we can find the corresponding $\theta_k^0$ from~\eqref{eq:3d_squared_sum} or~\eqref{eq:3d_rank2_first}.
Valid solutions always generate two pairs of wall normals:
$\{ \theta_k^0, \varphi_k^0 \}_{k=1}^2 = \{ (\theta_1^0, \varphi_1^0),  (-\theta_1^0, \varphi_1^0 + \pi) \}$ and $\{ \theta_k^0, \varphi_k^0 \}_{k=3}^4 = \{ (\theta_3^0, \varphi_3^0), (-\theta_3^0, \varphi_3^0 + \pi)\}$. Therefore, each reference room is made of two arbitrarily chosen walls and two walls parallel to them, resulting in parallelepipeds without its two bases.

As the case of $A^2 + B^2 = 0$ results in rather different geometries, it is analyzed separately in Section~\ref{sec:3drank2}.

\item \textit{Equivalent rooms:}
The corresponding angles in the equivalent room are computed from~\eqref{eq:3d_rank2_definition},
\begin{align}
\begin{aligned}
\label{eq:3d_rank2_transformation}
    \theta_k = \pi/2, \ \
    \varphi_k = g(\theta_k^0, \varphi_k^0, c,d,f) + s_k \pi,
\end{aligned}
\end{align}
where $s_k \in \{ 0,1\}$ and
\begin{align}
g(\theta_k^0, \varphi_k^0, c,d,f) = \arctan \frac{f \sin \theta_k^0 \sin \varphi_k^0}{c \sin \theta_k^0 \cos \varphi_k^0 + d \sin \theta_k^0 \sin \varphi_k^0}.
\end{align}

\item \textit{Equivalence class:}
We can set two wall orientations of a reference room by arbitrarily choosing $\varphi_1^0$ and $\varphi_3^0$, and by computing $\theta_1^0$ and $\theta_3^0$ from~\eqref{eq:3d_rank2_first}. 
By solving the system of two equations~\eqref{eq:3d_ABC} with $k \in \{1, 3\}$, we fix two parameters (e.g., $c$ and $d$) and leave the third parameter (e.g., $f$) free to generate new rooms equivalent to the reference. Walls parallel to those defined by $(\theta^0_1, \varphi^0_1)$ and $(\theta^0_3, \varphi^0_3)$ are determined by $(-\theta^0_1, \varphi^0_1 + \pi)$ and $(-\theta^0_3, \varphi^0_3 + \pi)$. Recall that the solutions of~\eqref{eq:3d_ABC} always come in pairs $(\theta_k^0, -\theta_k^0)$ and $(\varphi_k^0, \varphi_k^0 + \pi)$, so adding walls parallel to the two fixed ones does not violate~\eqref{eq:3d_ABC}.

As usual, we can choose $\vq^0 \in \mathbb{R}^K$ arbitrarily, and define an equivalence class of rooms generated by $\{\mathcal{P}^0_k\}_{k=1}^K$ as
\begin{align}
    \Big[ \{\mathcal{P}^0_k\}_{k=1}^K \Big] = \Big\{  \{\mathcal{P}_k\}_{k=1}^K \, \big| \, \varphi_k = g(\theta_k^0, \varphi_k^0, c,d,f) + s_k \pi,  &\nonumber \\ \, \theta_k = \pi/2, s_k \in \{ 0,1\}, q_k = q^0_k, f \in \mathbb{R}, &\nonumber \\ \, c, d \in \mathbb{R} \mbox { s.t.}~\eqref{eq:3d_ABC} \mbox{ satisfied}, \text{ for } 1 \leq k \leq K &\Big\}.
\end{align}

\item \textit{Corresponding trajectories:}
The nullspace of $\widebar{\mN}^\top$ is spanned by four vectors,
\begin{align*}
    \vv_1 = \begin{bmatrix}
    -a \\ -b \\ 1 \\ 0 \\ 0 \\ 0
    \end{bmatrix}, 
    \vv_2 = \begin{bmatrix}
    -c \\ -d \\ 0 \\ 1 \\ 0 \\ 0
    \end{bmatrix},
    \vv_3 = \begin{bmatrix}
    0 \\ -f \\ 0 \\ 0 \\ 1 \\ 0
    \end{bmatrix},
    \vv_4 = \begin{bmatrix}
    0 \\ 0 \\ 0 \\ 0 \\ 0 \\ 1
    \end{bmatrix},
\end{align*} 
so the waypoints in $\widebar{\mR}$ are related as
\begin{align}
\label{eq:3d_rank_2_waypoints}
    \begin{bmatrix}
    \vr_n^0 \\
    -\vr_n
    \end{bmatrix} = \gamma_1 \vv_1 + \gamma_2 \vv_2 + \gamma_3 \vv_3 + \gamma_4 \vv_4,
\end{align}
where $\gamma_1, \gamma_2, \gamma_3$ and $\gamma_4 \in \mathbb{R}$. It follows that the waypoints of the reference room are independent and can be chosen arbitrarily, whereas the corresponding waypoints of the equivalent rooms are given by~\eqref{eq:3d_rank_2_waypoints}.
An example is illustrated in Fig.~\ref{fig:case_2_3D}.
\end{enumerate}

\begin{figure}[H]
    \centering
    \includegraphics[width=\linewidth]{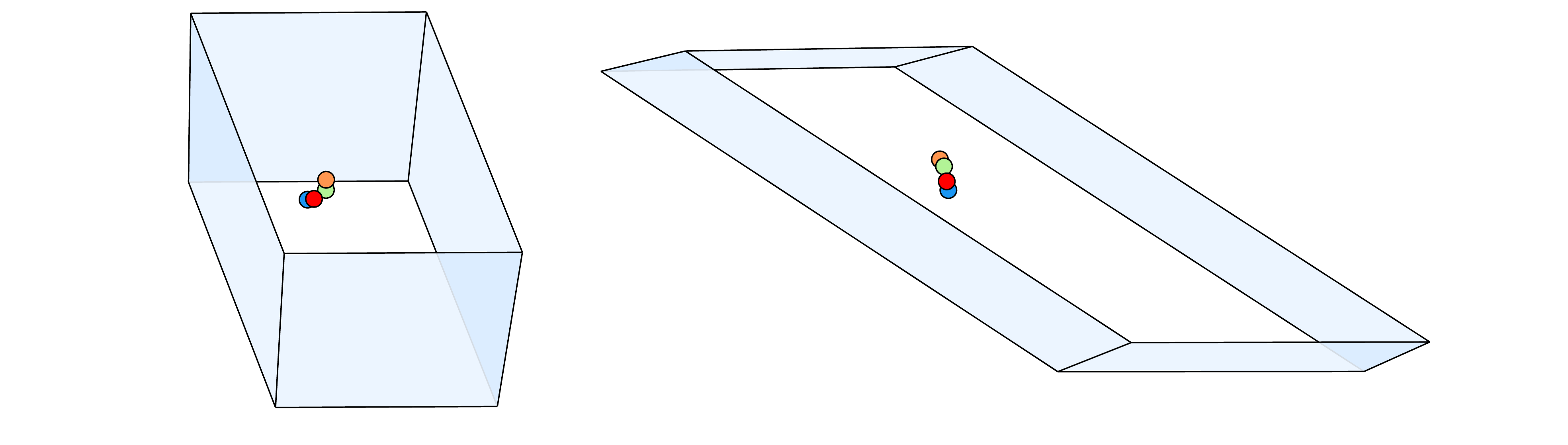}
    \caption{Two ``hollow parallelepipeds'' with the same PPDM.}
    \label{fig:case_2_3D}
\end{figure}

% ------------------------------------------------
\subsection{3D rank-2: Prisms without bases}
\label{sec:3drank2}

In step 3 of the previous case, we studied $A^2 + B^2 \neq 0$. Now we focus on $A^2 + B^2 = 0$ and omit steps $1$, $2$ and $6$ as they are identical to Section~\ref{sec:3drank2-I}.

\begin{enumerate}[leftmargin=*, wide, labelwidth=!, labelindent=0pt]
  \setcounter{enumi}{2}

\item \textit{Reference room:}
The case of $A^2 + B^2 = 0$ leads to $A = B = C=0$ and~\eqref{eq:3d_ABC} being satisfied for any value of $\varphi_k^0$.
By solving $A = B = C = 0$, we find explicit expressions for three dependent parameters in $\mT$,

\begin{align}
\begin{aligned}
\label{eq:3d_rank2_cdf}
c = \pm \sqrt{a^2+1},\
d = \pm \frac{ab}{\sqrt{a^2+1}}, \
f = \pm \sqrt{b^2 - \frac{a^2 b^2}{a^2+1} + 1}.
\end{aligned}
\end{align}

Then, from arbitrarily chosen angles $\varphi_k^0$, and the parameters in $\mT$ that satisfy~\eqref{eq:3d_rank2_cdf}, we compute $\theta_k^0$ from~\eqref{eq:3d_squared_sum} or~\eqref{eq:3d_rank2_first}. Such a room consists of $K$ walls parallel to a fixed line; this means that every triplet of walls forms a prismatic surface, or equivalently, every wall intersects the other two along lines.

To see this, observe that the rank of the coefficient matrix $\mN^0$ is $2$, while the rank of the augmented matrix $\mM^0$,
\begin{align}
\label{eq:augmented_matrix}
{\mM^0}^\top=
\begin{bmatrix}
    {\mN^0}^\top & \vq
\end{bmatrix},
\end{align}
can be $2$ or $3$. Indeed, the coefficient matrix from~\eqref{eq:3d_rank2_first} is
\begin{align}
    \vn_k^0 =
    \frac{1}{\sqrt{1 + (a \cos \theta_k^0 + b \sin \theta_k^0)^2}}
    \begin{bmatrix}
    \cos \varphi_k^0 \\
    \sin \varphi_k^0 \\
    a \cos \varphi_k^0 + b \sin \varphi_k^0
    \end{bmatrix}.
\end{align}
The third row of ${\mN^0}^\top$ is a linear combination of the first two rows so $\rank(\mN^0) = 2$.
From \eqref{eq:augmented_matrix} it follows that $\rank (\mM^0) = 3$, except for a set of $\vq$ of Lebesgue measure zero. A specific case of $\rank({\mM^0}) = 2$ occurs when the values of $\vq$ are chosen so that all walls intersect in one line. 

\item \textit{Equivalence rooms:}
The angles of the equivalent room $\theta_k$ and $\varphi_k$ are computed from~\eqref{eq:3d_rank2_transformation}.
We show that the equivalent room is a rotated version of the reference room.

The rotation ambiguity exists despite the reparametrization in step 2 because the normals in any equivalent room lie in a plane (the $xy$-plane in the reference room).
Then, transformation of the normals of $\{ \mathcal{P}^0_k \}_{k=1}^K$ to those of $\{ \mathcal{P}_k \}_{k=1}^K$ is determined by two angles, instead of three for a general rotation. We can factor any upper triangular matrix into a product of a rotation matrix around two axes and a square matrix by two Givens rotations \cite{givens1958computation}. Thus, $\mT$ being upper-triangular still allows for rotations specified by two angles. 

We introduce a matrix $\mR = (r_{ij})_{i,j=1}^3$ such that
\begin{align}
\begin{bmatrix}
\sin\theta_k \cos\varphi_k \\
\sin\theta_k \sin\varphi_k \\
\cos\theta_k 
\end{bmatrix} = \mR
\begin{bmatrix}
\sin\theta^0_k \cos\varphi^0_k \\
\sin\theta^0_k \sin\varphi^0_k  \\
\cos \theta^0_k
\end{bmatrix} \text{ for } 1 \leq k \leq K.
\end{align}
Together with~\eqref{eq:3d_rank2_reparametrization}, we obtain
\begin{align}
\begin{aligned}
\label{eq:3d_rank2_equalities}
    c &= r_{11} + a r_{13}, \quad
    d = r_{12} + b r_{13}, \quad
    0 = r_{21} + a r_{23}, \\
    f &= r_{22} + b r_{23}, \quad
    0 = r_{31} + a r_{33}, \quad
    0 = r_{32} + b r_{33},
\end{aligned}
\end{align}
so we can rewrite $\mR$ as
\begin{align}
\begin{aligned}
\label{eq:3d_rank2_rotation}
    \mR &=
    \begin{bmatrix}
    c - a r_{13} & d - b r_{13} & r_{13} \\
    - a r_{23} & f - b r_{23} & r_{23} \\
    - a r_{33} & - b r_{33} & r_{33}
    \end{bmatrix} = \\ &=
    \begin{bmatrix}
    c & d & 0 \\
    0 & f & 0 \\
    0 & 0 & 0
    \end{bmatrix}  - \begin{bmatrix}
    r_{13} \\ r_{23} \\ r_{33}
    \end{bmatrix} \begin{bmatrix}
    a & b & -1
    \end{bmatrix}.
\end{aligned}
\end{align}
To see that $\mR$ is a rotation, note that from $A = B = C = 0$, \eqref{eq:3d_rank2_rotation}, and \eqref{eq:3d_rank2_definitionsABC}, the columns of $\mR$ are orthonormal. 

% As a consequence,
% \begin{align}
% \begin{aligned}
% \label{eq:3d_rank2_cdi}
%     c r_{13} = a, \ \ \ \ \ \ 
%     d r_{13} + f r_{23} = b, \ \ \  \ \ \ 
%     r_{13}^2 + r_{23}^2 + r_{33}^2 = 1,
% \end{aligned}
% \end{align}
% from which we can determine the values of $r_{13}, r_{23}$ and $r_{33}$ as constants dependent on the parameters in $\mT$.
% By plugging them back into \eqref{eq:3d_rank2_rotation}, we can show that the rows of $\mR$ are orthonormal and $\det{\mR} = 1$, which is a sufficient condition for $\mR$ to be a rotation matrix. As the algebra is straightforward but long, we omit the details of its last step.
\end{enumerate}

The dependence of the corresponding waypoints is given in~\eqref{eq:3d_rank_2_waypoints} with an additional constraint on the parameters in~\eqref{eq:3d_rank2_cdf}. Intuitively, any waypoint that lies on a line parallel to walls generates the same PPDM. Thus the two equivalent rooms in the $\rank$-$2$ case in $3$D with $A^2 + B^2 = 0$ have identical geometries, but could have different waypoints lying on a line parallel to all walls, see Fig.~\ref{fig:case_20_3D}.
\begin{figure}[H]
    \centering
    \includegraphics[width=\linewidth]{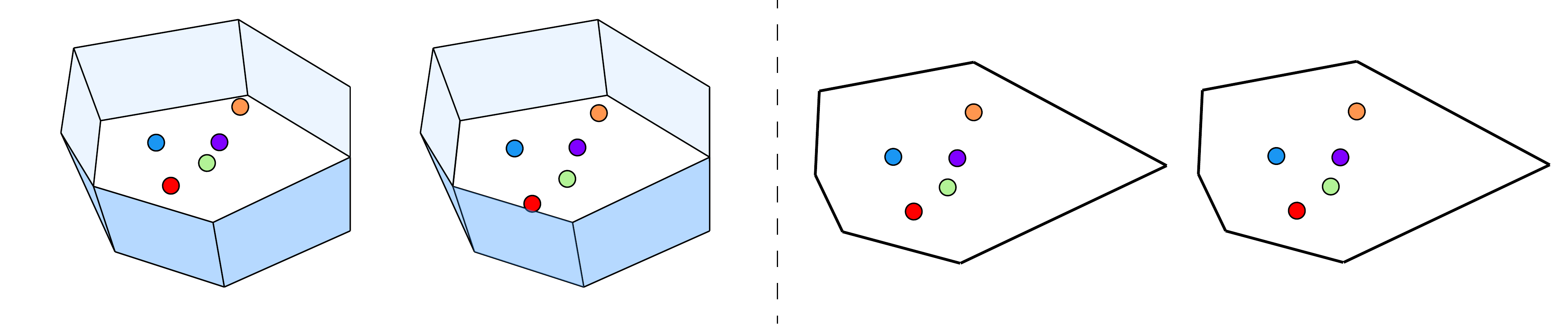}
    \begin{subfigure}{0.49\linewidth}
    \centering
    \caption{ }
    \end{subfigure}
    \begin{subfigure}{0.49\linewidth}
    \centering
    \caption{ }
    \end{subfigure}
    \caption{Two ``hollow prisms'' with the same PPDM. (a) The rooms are identical, but the waypoints differ. (b) A bird's eye view. The configurations from this angle seem identical.}
    \label{fig:case_20_3D}
\end{figure}

% ------------------------------------------------
\subsection{3D rank-3: Miscellaneous geometries}
\label{sec:3d_rank3_0}
\begin{enumerate}[leftmargin=*, wide, labelwidth=!, labelindent=0pt]

\item \textit{Linear dependence:}
The practically relevant shoebox rooms generate configurations not uniquely determined by PPDMs. For $\rank(\widebar{\mN}) = 3$,
\begin{align}
\begin{aligned}
\label{eq:3d_rank3_definition}
\begin{bmatrix}
\sin\theta_k \cos\varphi_k  \\ \sin\theta_k \sin\varphi_k \\ \cos\theta_k
\end{bmatrix} =
\mT
\begin{bmatrix}
\sin\theta^0_k \cos\varphi^0_k  \\ \sin\theta^0_k \sin\varphi^0_k \\ \cos\theta^0_k
\end{bmatrix},
\end{aligned}
\end{align}
where
\begin{align}
\mT = 
\begin{bmatrix}
a & b & c \\
d & e & f \\
g & h & i
\end{bmatrix}.
\end{align}

\item \textit{Reparametrization:}
As usual, we make $\mT$ upper triangular matrix by setting $d, g$ and $h$ to $0$.

\item \textit{Reference room:}
Since in~\eqref{eq:3d_rank3_definition} we have three equations with four angles for every wall $k$, we can express $\theta^0_k, \theta_k$ and $\varphi_k$ in terms of an arbitrarily chosen angle $\varphi^0_k$ and the parameters in $\mT$.
Squaring and summing \eqref{eq:3d_rank3_definition} gives
\begin{equation}
    \begin{aligned}
\label{eq:3d_rank3_theta0}
0 = \,&(A^2 + B^2) \sin^2(2 \theta_k^0) + 2(A + 2C) B \sin(2 \theta_k^0) \\&+ 4C (A + C),
\end{aligned}
\end{equation}
where
\begin{equation}
\begin{aligned}
A = \,&a^2 \cos^2\varphi^0_k + (b^2 + e^2)   \sin^2\varphi^0_k \\ &+ 2 a b \sin\varphi^0_k \cos\varphi^0_k - C-1, \\ 
B = \,&a c \cos\varphi^0_k + (b c + e f) \sin\varphi^0_k, \\
C = \,&c^2 + f^2 + i^2 - 1.      
\end{aligned}
\end{equation}
To find $\theta_k^0$, we solve~\eqref{eq:3d_rank3_theta0} and obtain
\begin{align}
\label{eq:3d_rank3_cos}
    \cos(2 \theta_k^0) = x_1 \quad \mbox{or} \quad \cos(2 \theta_k^0) = x_2,
\end{align}
with
\begin{align}
\label{eq:3d_rank3_x12}
    x_{1,2} = \frac{A(A+2C) \pm B\sqrt{B^2 - 4AC - 4 C^2}}{A^2 + B^2}.
\end{align}

We first consider $A^2 + B^2 \neq 0$, while the case of $A^2 + B^2 = 0$ is analyzed separately in Section~\ref{sec:3drank3II}.
Analogously to the $\rank$-$2$ case in $2$D or $3$D, not all solutions to~\eqref{eq:3d_rank3_cos} satisfy~\eqref{eq:3d_rank3_definition}; the four valid values of $\theta_k^0$ are identified by verifying
\begin{align}
\label{eq:3d_rank3_quadratic}
    1 = (a \sin \theta_k^0 \cos \varphi_k^0 + b \sin \theta_k^0 \sin \varphi_k^0 + c \cos \theta_k^0)^2&  \nonumber \\+ (e \sin \theta_k^0 \sin \varphi_k^0 + f \cos \theta_k^0)^2 + i^2 \cos^2 \theta_k^0&
\end{align}
for $1 \leq k \leq K$.
Contrary to the $\rank$-$2$ case in $2$D or $3$D, the values of $A, B$ and $C$ in~\eqref{eq:3d_rank3_cos} depend on $\varphi_k^0$ and the solutions to~\eqref{eq:3d_rank3_cos} vary for different walls $k$.
We denote them $\theta^0_{k,1}$, $\theta^0_{k,2}$, $\theta^0_{k,3}$ and $\theta^0_{k,4}$, where $\theta^0_{k,1}$ and $\theta^0_{k,2}$ are computed from $x_1$, while $\theta^0_{k,3}$ and $\theta^0_{k,4}$ from $x_2$. Moreover, they satisfy $\theta^0_{k,2} = \theta^0_{k,1} + \pi$ and $\theta^0_{k,4} = \theta^0_{k,3} + \pi$.

There are infinitely many ways to arrange the walls of the reference room, since for a fixed value of $\varphi_k^0$ there are four values of $\theta_k^0$ that satisfy~\eqref{eq:3d_rank3_quadratic}.
On the one hand, for a chosen $\varphi_k^0$ we can pick only one value $\theta_{k, j_k}^0$, $1 \leq j_k \leq 4$, and create wall normals $\{\theta_{k, j_k}^0, \varphi_k^0\}_{k=1}^K$.
Such rooms have different angles for every wall.
On the other hand, some rooms can have one value $\varphi_k^0$ associated to four walls, $\{\theta_{k,1}^0, \varphi_k^0\}_{k=1}^{K/4}$, $\{\theta_{k,2}^0, \varphi_k^0\}_{k=1}^{K/4}$, $\{\theta_{k,3}^0, \varphi_k^0\}_{k=1}^{K/4}$ and $\{\theta_{k,4}^0, \varphi_k^0\}_{k=1}^{K/4}$. As they lead to different shapes and many of them are common in real-world environments, they merit further analysis. The transformation to equivalent rooms is the same for all reference rooms, so we first define the classes, and then focus on reference rooms.

\item \textit{Equivalent rooms:}
We find the equivalent room from~\eqref{eq:3d_rank3_definition},
\begin{align}
\label{eq:3d_rank3_transformation}
\theta_k = t_k f(\theta^0_k, i), \quad \quad
\varphi_k = g(\theta^0_k, \varphi^0_k, \mT) + s_k \pi, 
\end{align}
where
\begin{equation}
\begin{aligned}
f(\theta^0_k, i) &= \arccos f \cos\theta^0_k,\\
g(\theta^0_k, \varphi^0_k, \mT) &= \arctan \frac{e \sin\theta^0_k \sin\varphi^0_k + f \cos\theta^0_k}{ \sin\theta^0_k (a\cos\varphi^0_k + b\sin\varphi^0_k) + c \cos\theta^0_k}, 
\end{aligned}
\end{equation}
$t_k \in \{ -1,1\}$ and $s_k \in \{ 0,1\}$.
The choice of $t_k$ uniquely determines $s_k$, such that~\eqref{eq:3d_rank3_definition} is satisfied.

\item \textit{Equivalence class:}
An equivalence class of rooms with respect to PPDM is given as
\begin{align}
\Big[ \{\mathcal{P}^0_k\}_{k=1}^K \Big] = \Big\{  \{\mathcal{P}_k\}_{k=1}^K \, \big| \, \varphi_k = g(\theta^0_k, \varphi^0_k, \mT) + s_k \pi, & \nonumber \\ \theta_k = t_k f(\theta^0_k, i),& \nonumber \\ a, b, c, e, f, i \in \mathbb{R} \mbox{ s.t.}~\eqref{eq:3d_rank3_quadratic} \mbox{ satisfied},& \nonumber \\ t_k \in \{ -1,1\}, s_k \in\{0,1\} \mbox{ s.t.}~\eqref{eq:3d_rank3_definition} \mbox{ satisfied},& \nonumber \\ q_k = q^0_k, \text{ for } 1 \leq k \leq K \Big\}&,
\end{align}
where $\{\mathcal{P}^0_k\}_{k=1}^K$ represents any of the reference rooms below.

The angles of the reference room $\{\varphi_k^0\}_{k=1}^K$ are chosen from $[0, 2\pi)$, while $\{\theta_k^0\}_{k=1}^K$ are computed from~\eqref{eq:3d_rank3_cos} and~\eqref{eq:3d_rank3_x12}.
For any arbitrarily chosen $\varphi_k^0$, we obtain four valid solutions $\theta_{k,1}^0 \hdots, \theta_{k, 4}^0$, which allow us to create up to four different walls for one fixed value of $\varphi_k^0$. We denote the number of walls created from one $\varphi_k^0$ by $\alpha$, and the number of independent walls in a room by $K_0$. Furthermore, we assume that we choose same $\alpha$ for all walls in a room, so we can categorize the reference rooms into four groups, from $\alpha=1$ to $\alpha = 4$.\footnote{We could also pick different $\alpha$ for every wall, but as such room construction only combines fundamental groups covered in the following and does not enrich our analysis, we do not discuss it further.}

\begin{enumerate}
    \item $\alpha = 1$. For a fixed $\varphi_k^0$, we select only one of the four valid solutions, assign it to $\theta_k^0$ and define a wall normal $k$ with $\{\theta_k^0, \varphi_k^0 \}$. Every wall introduces a new constraint~\eqref{eq:3d_rank3_quadratic} on six parameters in $\mT$. Two room-trajectory configurations that correspond to this case are shown in Fig.~\ref{fig:case_3_3d}.
    \begin{figure}[H]
    \centering
    \includegraphics[width=\linewidth]{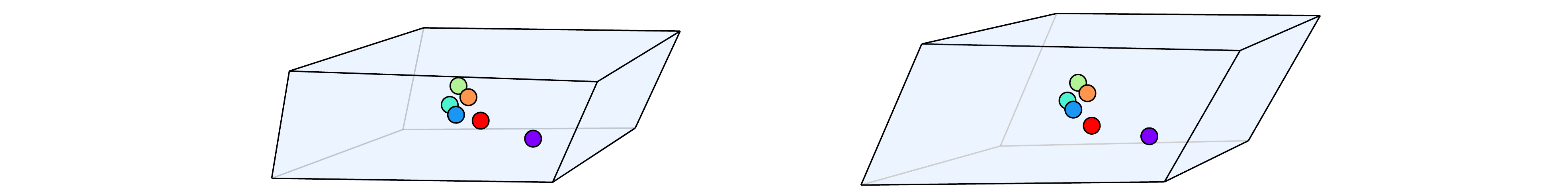}
    \caption{A pair of equivalent rooms in $3$D.}
    \label{fig:case_3_3d}
    \end{figure}
    
    \item $\alpha = 2$. For a fixed $\varphi_k^0$, we select two values of $\theta_k^0$ computed either from $x_1$ or $x_2$, and therefore create two parallel wall normals, for example $\{\theta_{k, 1}^0, \varphi_k^0 \}$ and $\{\theta_{k,2}^0, \varphi_k^0 \}$.
    The key observation in this case is that if~\eqref{eq:3d_rank3_quadratic} is satisfied for the wall normal $\{\theta_{k,1}^0, \varphi_k^0 \}$, then it is also satisfied for the wall normal $\{\theta_{k,2}^0, \varphi_k^0 \}$ without introducing additional constraints on the parameters in $\mT$. In Fig.~\ref{fig:case_3_3d_parallel} we illustrate an example of three rooms with $K_0= 3$.

    \begin{figure}[H]
        \centering
        \includegraphics[width=\linewidth]{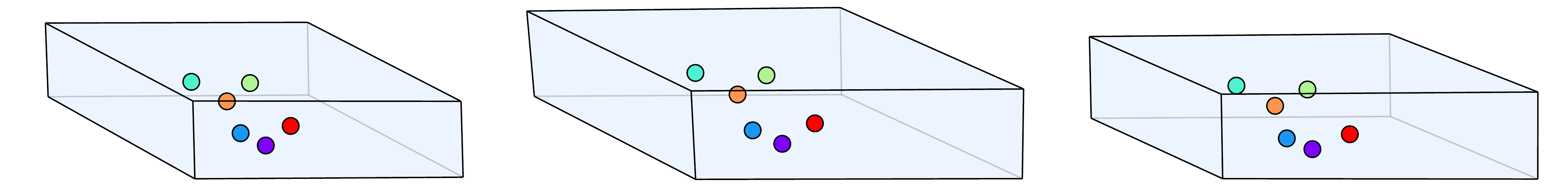}
        \caption{An example of equivalent rooms with three pairs of parallel walls.}
        \label{fig:case_3_3d_parallel}
    \end{figure}
    Similarly, for a fixed $\varphi_k^0$, we can select two values of $\theta_k^0$, one computed from $x_1$ and another one from $x_2$, and therefore create two differently oriented wall normals, for example $\{\theta_{k, 1}^0, \varphi_k^0 \}$ and $\{\theta_{k, 3}^0, \varphi_k^0 \}$.
    This case is equivalent to the one above; the only difference is that we replace one of the angles $\theta_{k,1}^0$ or $\theta_{k,2}^0$ with $\theta_{k,3}^0$ or $\theta_{k,4}^0$. Then, the reference room contains pairs of differently oriented, but dependent walls, instead of pairs of parallel walls.

    \item $\alpha = 3$. For a fixed $\varphi_k^0$, we select three values of $\theta_k^0$, one computed from $x_1$ ($x_2$) and two computed from $x_2$ ($x_1$). The room with $K = \alpha K_0$ walls constructed in such way contains $K_0$ pairs of parallel walls and $K_0$ variously oriented walls.
       
    \item $\alpha = 4$. For a fixed $\varphi_k^0$, we select all four valid solutions of $\theta_k^0$ and create two pairs of parallel wall normals, $\{\theta_{k,1}^0, \varphi_k^0 \}$, $\{\theta_{k,2}^0, \varphi_k^0 \}$, $\{\theta_{k,3}^0, \varphi_k^0 \}$ and $\{\theta_{k,4}^0, \varphi_k^0 \}$. Then, if~\eqref{eq:3d_rank3_quadratic} is satisfied for the wall normal $\{\theta_{k,1}^0, \varphi_k^0 \}$, it is also satisfied for the other three normals associated to $\varphi_k^0$ without introducing additional constraints on the parameters in $\mT$. Here, $K_0$ denotes the number of dependent quadruples of walls.

\end{enumerate}

In all of these cases, the reference room is defined by~\eqref{eq:3d_rank3_cos} and~\eqref{eq:3d_rank3_x12}. For $K = \alpha K_0 \geq 6 \alpha$ walls, the angles of the wall normals $\{\theta_k^0, \varphi_k^0\}_{k=1}^K$ cannot be chosen arbitrarily as they are determined by the parameters in $\mT$. Moreover, there is only one room equivalent to the reference, obtained from~\eqref{eq:3d_rank3_transformation}.
    
When $K  = \alpha K_0 < 6 \alpha$, we can choose \textit{any} room with $K$ walls to be the reference room and solve the system of $K_0$ equations~\eqref{eq:3d_rank3_quadratic} with $1 \leq \ k \leq K_0$ to find $K_0$ dependent parameters in $\mT$. Then, we generate new equivalent rooms from~\eqref{eq:3d_rank3_transformation} by changing the remaining $6 - K_0$ free parameters in $\mT$. An example of an arbitrarily chosen room with five walls together with the two equivalent rooms is shown in Fig.~\ref{fig:case_4}.

\begin{figure}[H]
    \centering
    \includegraphics[width=\linewidth]{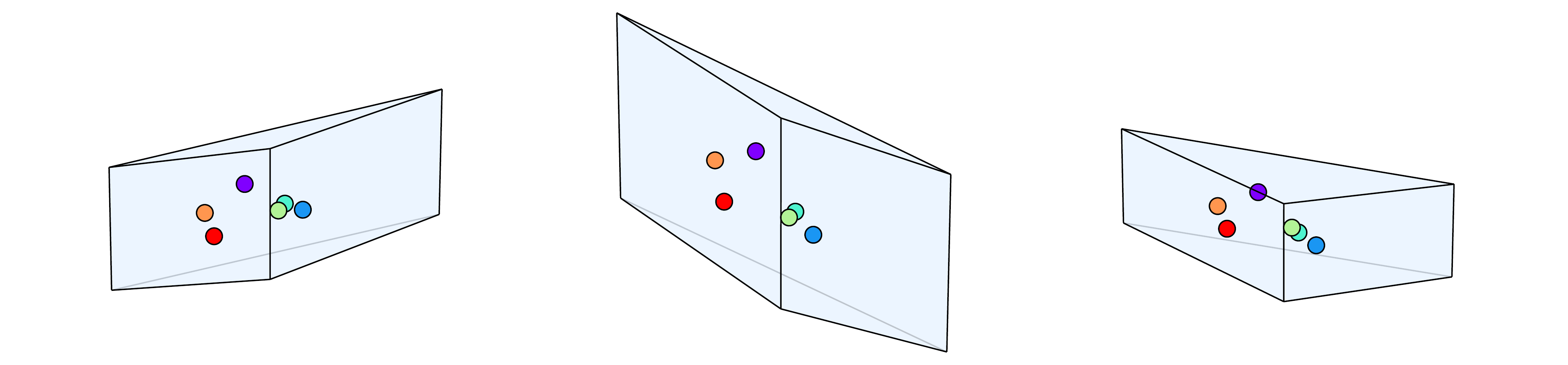}
    \caption{Rooms with less than six walls in $3$D that belong to the same equivalence class.}
    \label{fig:case_4}
\end{figure}

\item \textit{Corresponding trajectories:}
The nullspace of $\widebar{\mN}^\top$ is spanned by three vectors in all of the aforementioned cases,
\begin{align*}
\vv_1 &=
\begin{bmatrix}
 -a ,\ -b ,\ -c ,\ 1 ,\ 0 ,\ 0
\end{bmatrix}^\top, \\
\vv_2 &=
\begin{bmatrix}
 0 ,\ -e ,\ -f ,\ 0 ,\ 1 ,\ 0
\end{bmatrix}^\top,\\
\vv_3 &=
\begin{bmatrix}
0 ,\ 0 ,\ -i ,\ 0 ,\ 0 ,\ 1
\end{bmatrix}^\top.
\end{align*}
Then,
\begin{align}
\begin{bmatrix}
\vr_n^0 \\ -\vr_n
\end{bmatrix} =
\vv_1 \gamma_1 +
\vv_2 \gamma_2 +
\vv_3 \gamma_3,
\end{align}
where $\gamma_1, \gamma_2$ and $\gamma_3 \in \mathbb{R}$. The waypoints in one room are chosen arbitrarily and a non-rigid transformation $\mT^\top$ is applied to compute the waypoints in the equivalent room,
\(
\vr_n^0 = \mT^\top \vr_n.
\)

\end{enumerate}

\subsection{3D rank-3: Two sets of parallel walls}
\label{sec:3drank3II}

There is another equivalence class arising from $\rank(\widebar{\mN}) = 3$ for $A^2 + B^2 = 0$ and $\cos \varphi_k^0 \neq 0$.
One can show that these constraints lead to rooms with arbitrarily chosen angles $\theta_k^0$ and constant values for $\varphi_k^0$ (up to a shift by $\pi$), i.e., rooms with all walls parallel to a line.
An analysis similar to that in Section~\ref{sec:3drank2} shows that the rooms in the same equivalence class are simply rotated versions of the reference room.

\begin{enumerate}[leftmargin=*, wide, labelwidth=!, labelindent=0pt]\addtocounter{enumi}{2}
\item \textit{Reference room:}
We continue with $A^2 + B^2 = 0$ which implies $A = B = C = 0$, and in addition we assume that $\cos \varphi_k^0 = 0$. We omit steps $1$, $2$ and $6$ as they are identical to Section~\ref{sec:3drank2}.
From $B=0$, it follows that
\begin{align}
\label{eq:3d_rank3_A_B_0_new}
    a c = 0 \quad \mbox{and} \quad bc + ef = 0.
\end{align}
From~\eqref{eq:3d_rank3_A_B_0_new}, we conclude that either $a \neq 0, c = 0$, or $a = 0, c \neq 0$, or $a = c = 0$. The last two cases are not of our interest as $a=0$ implies that the $x$ coordinates of $\vr^0_n$ are $0$, and the points lie in the $yz$-plane. Such a degenerate trajectory is covered in our next case, $\rank(\widebar{\mN}) = 4$, so we do not study it further here. A similar observation can be made for $a \neq 0, c = 0, e=0$; the $y$ coordinates of $\vr^0_n$ are proportional to their $x$ coordinates, so the points lie in a plane, which corresponds to $\rank(\widebar{\mN}) = 4$.

A new equivalence class arises for $a \neq 0, c = 0, f = 0$. From $C = 0$, we obtain that $i = \pm 1$, while $A = 0$ defines $\varphi_k^0$,
\begin{align}
\label{eq:3d_rank3_A_0_ce_0}
    (a \cos \varphi_k^0 + b \sin \varphi_k^0)^2 + e^2 \sin^2 \varphi_k^0 = 1.
\end{align}
By introducing $u = \tan \frac{\varphi_k^0}{2}$ and $z = \frac{u^2 - 1}{u}$, we can find the solutions of~\eqref{eq:3d_rank3_A_0_ce_0} in terms of $z$,
\begin{align}
    z_{1, 2} = \frac{2 a b \pm 2 \sqrt{-a^2 e^2 + a^2 + b^2 - e^2 - 1}}{a^2-1},
\end{align}
from which we can express the four solutions of $\varphi_k^0$,
\begin{align}
\begin{aligned}
\label{eq:3d_rank3_A_0_ce_0_varphi}
    \varphi_k^0 &= 2 \arctan{\frac{z_i \pm \sqrt{z_i^2 + 4}} {2}},
\end{aligned}
\end{align}
for $i \in \set{1, 2}$. We observe that the normals computed from $z_1$ generate rooms with walls parallel to a certain line $\ell_1$. Analogously, the normals generated by $z_2$ are parallel to another line $\ell_2$. Therefore, to construct the reference room, we choose $\{\theta_k^0 \}_{k=1}^{K}$ from $[0, \pi)$, while $\{ \varphi_k^0 \}_{k=1}^K$ are computed from~\eqref{eq:3d_rank3_A_0_ce_0_varphi}, such that $K_1$ wall normals are derived from $z_1$ and $K_2$ wall normals from $z_2$, where $ K_1 + K_2 = K$.

\item \textit{Equivalent room:}
We find the equivalent rooms from~\eqref{eq:3d_rank3_definition} by the same computations as in Section~\ref{sec:3d_rank3_0}.

\item \textit{Equivalence class:}
The equivalence class also corresponds to the one in Section~\ref{sec:3d_rank3_0} with $c=0, f=0$ and $i = \pm 1$. The free parameter $a$ generates equivalent rooms,
\begin{align}
\Big[ \{\mathcal{P}^0_k\}_{k=1}^K \Big] = \Big\{ \{\mathcal{P}_k\}_{k=1}^K  \, \big| \, \theta_k = f(\theta_k^0, t_k, i=\pm 1),& \nonumber \\ \varphi_k = g(\theta_k^0, \varphi_k^0, s_k, a,b, c=0,e, f=0),& \nonumber \\ t_k, s_k \in\{0,1\} \mbox{ s.t. } \eqref{eq:3d_rank3_definition} \mbox{ satisfied},& \nonumber \\ b, e \in \mathbb{R} \mbox{ s.t. } \eqref{eq:3d_rank3_A_B_0_new} \mbox{ satisfied}, a \in \mathbb{R},& \nonumber \\ q_k = q_k^0, \mbox{ for } 1 \leq k \leq K \Big\}&.
\end{align}

\end{enumerate}
Note that the walls computed from $z_1$ do not have to enclose any specific shape, as long as they are equally inclined to all the walls obtained from $z_2$.

An interesting realistic room that belongs to this class is a room made up of four parallel walls that are perpendicular to the ceiling and the floor. By tilting the ceiling and the floor (changing the value of $a$), we can generate infinitely many equivalent rooms with respect to PPDM, see Fig.~\ref{fig:3d_rank3_2groups}.

\begin{figure}[H]
    \centering
    \includegraphics[width=\linewidth]{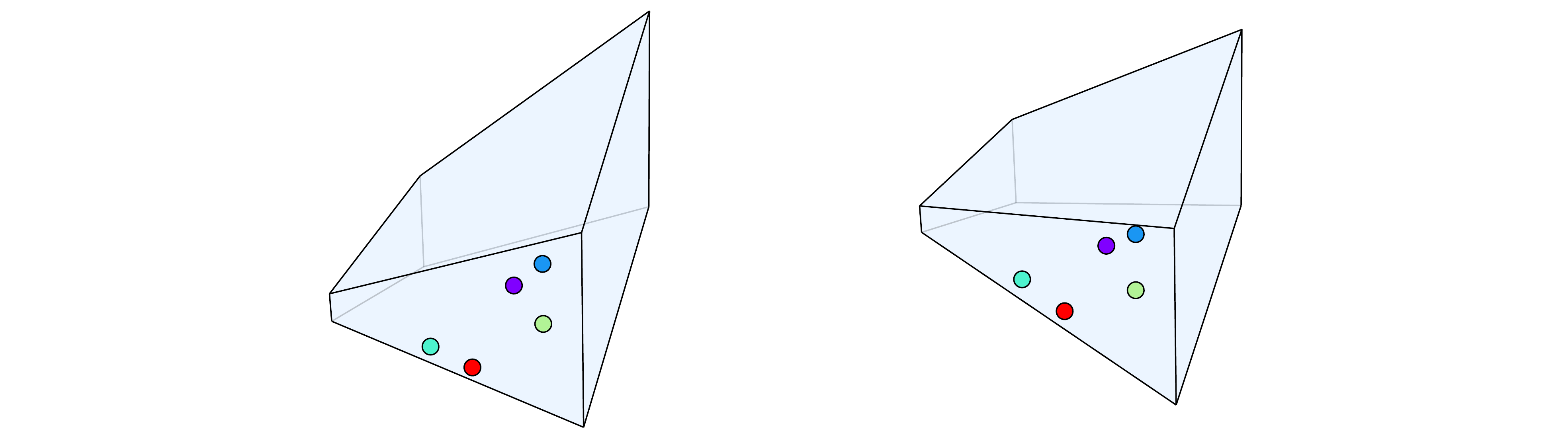}
    \caption{Equivalent rooms with two groups of walls enclosing a prismatic surface.}
    \label{fig:3d_rank3_2groups}
\end{figure}

% ------------------------------------------------
\subsection{3D rank-4: Planar trajectories}

\begin{enumerate}[leftmargin=*, wide, labelwidth=!, labelindent=0pt]

\item \textit{Linear dependence:}
To achieve $\rank(\widebar{\mN}) = 4$, we assume that the fourth and the fifth column of $\widebar{\mN}$ are linear combinations of the remaining four,
\begin{equation}
\begin{aligned}
\label{eq:3d_rank4_definition}
&\begin{bmatrix}
\sin\theta_k \cos\varphi_k \\
\sin\theta_k \sin\varphi_k \\
\end{bmatrix} =
\mT \begin{bmatrix}
\sin\theta^0_k \cos\varphi^0_k \\
\sin\theta^0_k \sin\varphi^0_k \\
\cos\theta^0_k \\
\cos\theta_k \\
\end{bmatrix},
\end{aligned}
\end{equation}
where
\begin{align}
\mT = \begin{bmatrix}
a & b & c & d \\
e & f & g & h
\end{bmatrix}.
\end{align}

\item \textit{Reparametrization:}
As $r > m$, we cannot rewrite~\eqref{eq:3d_rank4_definition} so that the normals of $\mathcal{R}$ and $\mathcal{R}^0$ are on different sides.

\item \textit{Reference room}: In~\eqref{eq:3d_rank4_definition} we have two equations with four unknown angles for every $k$. Since the system is underdetermined, we can choose $\left\{ \theta^0_k, \varphi^0_k \right\}_{k=1}^{K}$ arbitrarily.

\item \textit{Equivalent rooms}:
We solve~\eqref{eq:3d_rank4_definition} for $\theta_k$ and $\varphi_k$, and express their dependence on $\theta^0_k$, $\varphi^0_k$ and the parameters in $\mT$,
\begin{equation}
\begin{aligned}
    \theta_k = s_k f(\theta_k^0, \varphi_k^0, \mT), \quad \quad
    \varphi_k = t_k h(\theta_k^0, \varphi_k^0, \mT), 
\end{aligned} 
\end{equation}
where $s_k, t_k \in \{ -1, 1 \}$, and
\begin{equation}
\begin{aligned}
    f(\theta^0_k, \varphi^0_k, &\mT) = \arccos \frac{-d G_a - h G_e \pm \sqrt{G}}{1 + d^2 + h^2}, \\
    h(\theta^0_k, \varphi^0_k,& \mT) = \arccos \frac{d \cos \theta_k+ G_a}{\sin \theta_k},
\end{aligned}
\end{equation}
and we introduced the following shortcuts:
\begin{equation}
\begin{aligned}
    G_a &:= a \sin \theta_k^0 \cos \varphi_k^0 + b \sin \theta_k^0 \sin \varphi_k^0 + c \cos \theta_k^0,\\ 
    G_e &:= e \sin \theta_k^0 \cos \varphi_k^0 + f \sin \theta_k^0 \sin \varphi_k^0 + g \cos \theta_k^0,\\
    G &:= (d G_a + h G_e)^2 - (1 + d^2 + h^2)(G_a^2 + G_e^2 - 1).
\end{aligned}
\end{equation}

\item \textit{Equivalence class}:
An equivalence class of rooms with respect to PPDM is
\begin{align}
    \Big[ \{\mathcal{P}^0_k\}_{k=1}^K \Big] = \Big\{ \{\mathcal{P}_k\}_{k=1}^K  \, \big| \, \theta_k =  s_k f(\theta_k^0, \varphi_k^0, \mT),& \nonumber \\ \varphi_k = t_k h(\theta_k^0, \varphi_k^0, \mT), &\nonumber \\ a, b, c, d, e,f,g,h \in \mathbb{R}, s_k, t_k \in \{0,1\}, \nonumber \\ q_k=q_k^0 \mbox{ for } 1 \leq k \leq K \Big\}&.
\end{align}

\item \textit{Corresponding trajectories:}
The nullspace of $\widebar{\mN}^\top$ is spanned by two vectors
\begin{align*}
\vv_1 &=
\begin{bmatrix}
 -a ,\ -b ,\ -c ,\ 1 ,\ 0 ,\ -d
\end{bmatrix}^\top, \\
\vv_2 &=
\begin{bmatrix}
 -e ,\ -f ,\ -g ,\ 0 ,\ 1 ,\ -h
\end{bmatrix}^\top,
\end{align*}
so the $n$th row of $\widebar{\mR}$ is
\begin{align}
\begin{bmatrix}
\vr^0_n \\ -\vr_n
\end{bmatrix}=
\vv_1 \gamma_1 +
\vv_2 \gamma_2,
\label{eq:coplanar}
\end{align}
where $\gamma_1, \gamma_2 \in \mathbb{R}$.
From~\eqref{eq:coplanar} we have that one coordinate of the waypoints $\vr^0_n$ and $\vr_n$ is a linear combination of the remaining two, meaning that the waypoints lie in a plane.

\end{enumerate}

We conclude that for arbitrarily chosen wall normals of the reference room, we can always find another room with identical distance measurements, as long as the trajectories in both rooms are planar, as in Fig.~\ref{fig:case_2a_3D}.

\begin{figure}[H]
    \centering
    \includegraphics[width=\linewidth]{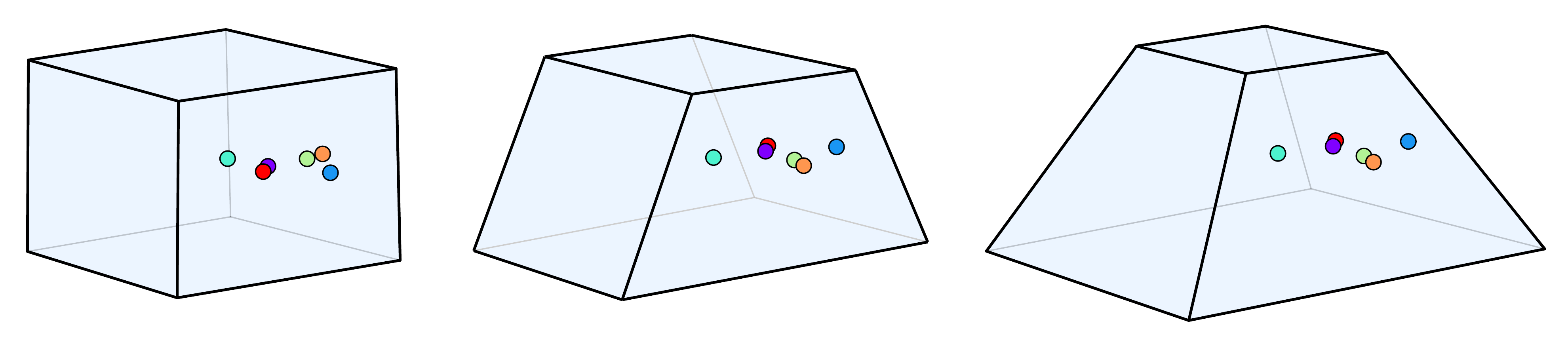}
    \caption{Rooms with planar trajectories and the same PPDM.}
    \label{fig:case_2a_3D}
\end{figure}

% ------------------------------------------------
\subsection{3D rank-5: Linear trajectories}

\begin{enumerate}[leftmargin=*, wide, labelwidth=!, labelindent=0pt]

\item \textit{Linear dependence:}
Finally, let $\rank(\widebar{\mN}) = 5$, so that one column of $\widebar{\mN}^\top$ is a linear combination of the remaining independent columns,
\begin{align}
\label{eq:3d_rank5_definition}
\cos\theta_k = 
\mT
\begin{bmatrix}
\sin\theta_k \cos\varphi_k \\
\sin\theta_k \sin\varphi_k \\
\sin\theta^0_k \cos\varphi^0_k \\
\sin\theta^0_k \sin\varphi^0_k \\
\cos\theta^0_k \\
\end{bmatrix}, \, \text{where} \, \mT =\begin{bmatrix}
a & b & c & d & e 
\end{bmatrix}.
\end{align}

\item \textit{Reparametrization:}
Since $r > m$, this step is a no-op.

\item \textit{Reference room:} From~\eqref{eq:3d_rank5_definition}, we can choose walls of the reference room arbitrarily.

\item \textit{Equivalent rooms:}
Furthermore, we can express $\theta_k$ as a function of $\varphi_k$, $\theta^0_k$, $\varphi^0_k$ and the parameters in $\mT$,
\begin{align}
    \theta_k = f(\varphi_k, \theta_k^0, \varphi_k^0, s_k, \mT),
\end{align}
where
\begin{equation}
\begin{aligned}
    h(\varphi_k, a, b) &= a \cos \varphi_k + b \sin \varphi_k, \\
    f(\varphi_k, \theta_k^0, \varphi_k^0, s_k, \mT) &= s_k g(\varphi_k, \theta_k^0, \varphi_k^0, \mT) - \arctan \!h(\varphi_k, a, b), \\
    g(\varphi_k, \theta_k^0, \varphi_k^0, \mT) &= (\cdots)\\
    & \hspace{-2.2cm} (\cdots) = \arccos \frac{d \sin \theta_k^0 \sin \varphi_k^0 + e \cos \theta_k^0 +c \sin \theta_k^0 \cos \varphi_k^0}{\sqrt{h(a, b, \varphi_k)^2 + 1}},
\end{aligned}
\end{equation}
with $s_k \in \{ -1, 1 \}$.

\item \textit{Equivalence class:}
An equivalence class of rooms with respect to PPDM is given by
\begin{align}
    \Big[ \{\mathcal{P}^0_k\}_{k=1}^K \Big] = \Big\{ \{\mathcal{P}_k\}_{k=1}^K  \, \big| \, \theta_k = f(\varphi_k, \theta_k^0, \varphi_k^0, s_k,\mT),& \nonumber \\ \varphi_k \in \left[0, 2\pi \right\}, a, b, c, d, e \in \mathbb{R}, \nonumber \\ \, s_k \in \{-1,1\}, q_k=q_k^0, \mbox{ for } 1 \leq k \leq K \Big\}&.
\end{align}

\item \textit{Corresponding trajectories:}
The nullspace of $\widebar{\mN}^\top$ is spanned by
\begin{align*}
\vv_1 = \begin{bmatrix} -c, & -d, & -e, & -a, & -b, & 1 \end{bmatrix}^\top,
\end{align*}
so the columns of $\widebar{\mR}$ have to be of the form
\begin{equation}
    \begin{bmatrix}
    \vr_n^0 \\
    -\vr_n
    \end{bmatrix} = \vv_1 \gamma,
\end{equation}
where $\gamma \in \mathbb{R}$. Therefore, $x$ and $y$ coordinates of the waypoints $\left\{ \vr^0_n \right\}_{n=1}^N$ and $\left\{ \vr_n \right\}_{n=1}^N$ are only scaled values of the $z$ coordinates of $\left\{ \vr_n \right\}_{n=1}^N$, so the trajectories are linear.

\end{enumerate}

We conclude that for any arbitrarily chosen room, we can always find another room with the same PPDM, as long as the trajectories in both rooms are linear. While linear trajectories may seem a special case of the previous one, the room transformations are rather different. One example of such configurations is illustrated in Fig.~\ref{fig:case_5_3D}.

\begin{figure}[H]
    \centering
    \includegraphics[width=\linewidth]{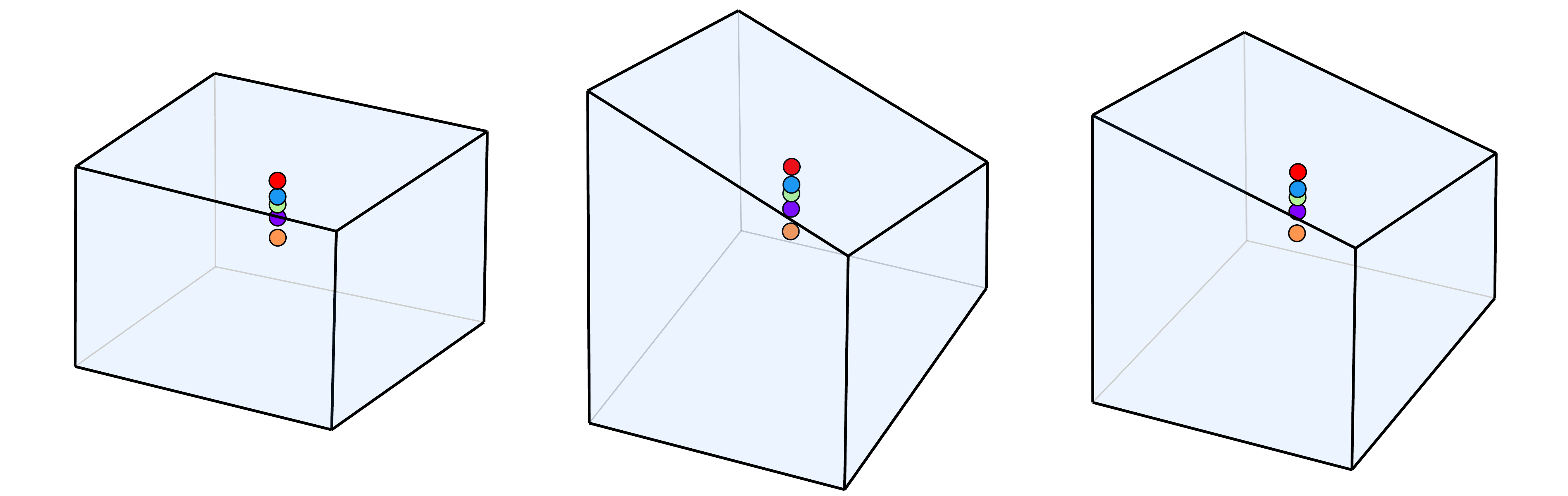}
    \caption{Rooms with linear trajectories and the same PPDM.}
    \label{fig:case_5_3D}
\end{figure}

\section{Conclusion}
\label{discussion}

We derived sufficient and necessary conditions for unique reconstruction of point--plane configurations from their pairwise distances. Our analysis hinges on a new algebraic tool called point-to-plane distance matrix (PPDM). We exhaustively identify the geometries of points and planes that cannot be distinguished given their PPDMs.

Our motivation comes from the challenging problem of multipath-based simultaneous localization and mapping (SLAM) and our study has consequences for practical indoor localization problems. Picture an unknown room with no preinstalled infrastructure and a mobile device equipped with a single omnidirectional source and a single omnidirectional receiver. The distance measurements between the points and planes are given as the time-of-flights of the first-order echoes recorded by the device.
Therefore, our theoretical results provide a fundamental understanding and constraints under which rooms can uniquely be reconstructed from only first-order echoes.

While our analysis here starts with the PPDM, preparing the PPDM in real scenarios puts forward additional challenges, namely PPDM completion and denoising, and echo sorting. Our ongoing research includes the development and implementation of computational tools and heuristics for localization from noisy, incomplete, and unlabeled PPDMs.

\begin{appendix}

For all $m$ and $r$ we worked with a particular selection of $r$ independent columns. We prove here that this choice can be made without loss of generality. We will call the particular column choice in Sections \ref{sec:classification_2d} and  \ref{sec:classification_3d} \textit{the original choice}.

First note that there is symmetry between reference and equivalent rooms. For example, for $r=1$ in $2$D, 
given the original choice of $r$ independent columns we have
\begin{align}
\label{eq:example_original}
\begin{bmatrix}
\sin \varphi_k,  \cos \varphi^0_k, \sin \varphi^0_k
\end{bmatrix}^\top = \begin{bmatrix}
a,  b,  c
\end{bmatrix}^\top
\cos \varphi_k.
\end{align}
We can swap the normals $\{\varphi^0_k \}_{k=1}^K$ and $\{\varphi_k \}_{k=1}^K$ for every $k$, and obtain a new, symmetric choice of $r$ independent columns
\begin{align}
\label{eq:example_swapped}
\begin{bmatrix}
\sin \varphi_k^0, \cos \varphi_k,  \sin \varphi_k
\end{bmatrix}^\top = \begin{bmatrix}
a, b, c
\end{bmatrix}^\top
\cos \varphi_k^0.
\end{align}
The two systems \eqref{eq:example_original} and \eqref{eq:example_swapped} give the same equivalence class. 

A similar conclusion follows if the new choice is obtained by rearranging the order of the coordinates of the normals.
Again, for $r=1$ in $2$D we have that
\begin{align}
\label{eq:2d_rank_1_swapped}
\begin{bmatrix}
\cos \varphi_k^0, \cos \varphi_k, \sin \varphi_k
\end{bmatrix}^\top = \begin{bmatrix}
a, b, c
\end{bmatrix}^\top
\sin \varphi_k^0
\end{align}
can be transformed to the studied case of~\eqref{eq:2d_rank_1}. Indeed, by applying a rotation by $\pi/2$
to the normals of the reference room, we obtain a new reference room which satisfies~\eqref{eq:2d_rank_1}, but rotated configurations are considered to be equivalent.

In the following we show that any choice of $r$ independent columns not covered by the two previous examples can be transformed into one of the cases analyzed in Sections~\ref{sec:classification_2d} and~\ref{sec:classification_3d} (for $r=2$ in $2$D and $r \in \{2, 3, 4\}$ in $3$D).

\begin{enumerate}[leftmargin=*, wide, labelwidth=!, labelindent=0pt]

\item \textbf{2D rank-2.} 
By symmetry, it is sufficient to show that
\begin{align}
    \begin{bmatrix}
    \cos \varphi_k \\ \cos \varphi_k^0
    \end{bmatrix} =
    \begin{bmatrix}
    a & b \\ c & d
    \end{bmatrix}
    \begin{bmatrix}
    \sin \varphi_k \\ \sin \varphi_k^0
    \end{bmatrix}
\end{align}
can be transformed into~\eqref{eq:2d_rank2_definition}.
For $c \neq 0$, it follows directly:
\begin{align}
    \begin{bmatrix}
    \cos \varphi_k \\ \sin \varphi_k
    \end{bmatrix} =
    \frac{1}{c}
    \begin{bmatrix}
    a & bc-ad \\ 1 & -d
    \end{bmatrix}
    \begin{bmatrix}
    \cos \varphi_k^0 \\ \sin \varphi_k^0
    \end{bmatrix}.
\end{align}
For $c = 0$ we have $\tan \varphi_k^0 = \tfrac{1}{d}$, addressed in \eqref{eq:2d_rank_1}.

\item \textbf{3D rank-2.} 
By symmetry, we only analyze
\begin{align}
\label{eq:3d_rank2_appendix_new_1}
\begin{bmatrix}
\cos \theta_k^0 \\ \sin \theta_k^0 \sin \varphi_k^0 \\ \sin \theta_k \sin \varphi_k \\ \cos \theta_k 
\end{bmatrix} = 
\begin{bmatrix}
a & b \\ c & d \\ e & f \\ g & h
\end{bmatrix}
\begin{bmatrix}
\sin \theta_k^0 \cos \varphi_k^0 \\ \sin \theta_k \cos \varphi_k
\end{bmatrix}
\end{align}
and transform it into~\eqref{eq:3d_rank2_definition} as
\begin{align}
\label{eq:3d_rank2_appendix_new}
\begin{bmatrix}
\cos \theta_k^0 \\ \sin \theta_k \cos \varphi_k \\ \sin \theta_k \sin \varphi_k \\ \cos \theta_k 
\end{bmatrix} = \frac{1}{d}
\begin{bmatrix}
ad-bc & b \\ -c & 1 \\ ef-cf & f \\ gd-ch & h
\end{bmatrix}
\begin{bmatrix}
\sin \theta_k^0 \cos \varphi_k^0 \\ \sin \theta_k^0 \sin \varphi_k^0
\end{bmatrix}
\end{align}
for $d\neq 0$.
If $d=0$ and $b \neq 0$, a substitution $\sin\theta_k\cos\varphi_k = \tfrac{1}{b}(\cos\theta_k^0 - a \sin\theta_k^0\cos\theta_k^0)$ from the first equation of~\eqref{eq:3d_rank2_appendix_new_1} into the last two equations of~\eqref{eq:3d_rank2_appendix_new_1} gives a system equivalent to~\eqref{eq:3d_rank2_definition}.
For $b=d=0$, we get constant normals, discussed in \eqref{eq:3d_rank1_a}.

\item \textbf{3D rank-3.} 
Again, we only analyze
\begin{align}
\label{eq:3d_rank3_appendix_new}
\begin{bmatrix}
\sin \theta_k \cos\varphi_k \\ \sin \theta_k \sin \varphi_k  \\ \cos \theta_k^0
\end{bmatrix} = 
\begin{bmatrix}
a & b & c \\ d & e & f \\ g & h & i
\end{bmatrix}
\begin{bmatrix}
\sin \theta_k^0 \cos\varphi_k^0 \\ \sin \theta_k^0 \sin \varphi_k^0  \\ \cos \theta_k
\end{bmatrix}
\end{align}
and show that we can transform it into~\eqref{eq:3d_rank3_definition}. Indeed, for $i \neq 0$,
\begin{align}
\begingroup 
\setlength\arraycolsep{3.3pt}
\begin{bmatrix}
\sin \theta_k \cos\varphi_k \\ \sin \theta_k \sin \varphi_k  \\
\end{bmatrix} = \frac{1}{i}
\begin{bmatrix}
ai-cg & bi-ch & c \\ di-fg & ei-fh & f \\ -g & -h & 1
\end{bmatrix}
\begin{bmatrix}
\sin \theta_k^0 \cos\varphi_k^0 \\ \sin \theta_k^0 \sin \varphi_k^0  \\ \cos \theta_k^0
\end{bmatrix}.
\endgroup
\end{align}
For $i=0, g\neq0$ or $i=g=0, h\neq0$ we can substitute either $\sin\theta_k^0\cos\varphi_k^0$ or $\sin\theta_k^0\sin\varphi_k^0$ from the last equation of~\eqref{eq:3d_rank3_appendix_new} into the first two equations of~\eqref{eq:3d_rank3_appendix_new}, getting \eqref{eq:3d_rank4_definition}.
The same holds for $i=g=h=0$, with an additional constraint $\cos\theta_k^0=0$ on the reference normals.

\item \textbf{3D rank-4.}
Let us assume
\begin{align}
\label{eq:3d_rank4_appendix_new}
\begin{bmatrix}
\sin\theta_k \cos\varphi_k \\
\sin\theta^0_k \cos\varphi^0_k \\
\end{bmatrix} =
\begin{bmatrix}
a & b & c & d \\ e & f & g & h
\end{bmatrix}
\begin{bmatrix}
\sin\theta_k \cos\varphi_k \\
\sin\theta^0_k \sin\varphi^0_k \\
\cos\theta^0_k \\
\cos\theta_k \\
\end{bmatrix}.
\end{align}
Thanks to symmetry, this is the only case of our interest and we transform it to the well-studied system~\eqref{eq:3d_rank4_definition} for $e \neq 0$:
\begin{align*}
\begin{bmatrix}
\sin\theta_k \cos\varphi_k \\
\sin\theta_k \sin\varphi_k \\
\end{bmatrix} = \frac{1}{e}
\begin{bmatrix}
a & 1 \\ b e - a f & -f \\ c e - a g & -g \\ d e - a h & -h 
\end{bmatrix}^\top
\begin{bmatrix}
\sin\theta^0_k \cos\varphi^0_k \\
\sin\theta^0_k \sin\varphi^0_k \\
\cos\theta^0_k \\
\cos\theta_k \\
\end{bmatrix}.
\end{align*}
If $e = 0$ and $h \neq 0$, substituting $\cos \theta_k$ from the second into the first equation of~\eqref{eq:3d_rank4_appendix_new} gives~\eqref{eq:3d_rank4_definition}. By similar substitutions for $e=h=0$, $f \neq 0$, and $e=f=h=0$, $g\neq 0$, we get \eqref{eq:3d_rank5_definition}. Finally, $e=f=h=g=0$ also corresponds to $r=5$ in $3$D, with an additional constraint  $\sin\theta_k^0\cos\varphi_k^0=0$.

\end{enumerate}

\end{appendix}

\bibliographystyle{IEEEtran}
\bibliography{library}



% \begin{IEEEbiography}{Miranda Krekovi\'c}
% \end{IEEEbiography}

% \begin{IEEEbiography}{Ivan Dokmani\'c}
% \end{IEEEbiography}

% \begin{IEEEbiography}{Martin Vetterli}
% \end{IEEEbiography}

%\newpage


%\vfill


\end{document}